\pdfoutput=1

\documentclass[sigconf, nonacm]{acmart}
\usepackage{pvldb}

\renewcommand\vldbavailabilityurl{https://github.com/S1mpleCod/Near-optimal-approximation-for-DPPS}

\usepackage{amsmath,amssymb,amsthm}
\usepackage{mathtools}
\usepackage{booktabs}
\usepackage{array}
\usepackage{float}
\usepackage{enumitem}
\usepackage{xspace}
\usepackage{microtype}
\usepackage[linesnumbered,ruled,vlined]{algorithm2e}
\usepackage{tikz}
\usetikzlibrary{arrows.meta,positioning,calc,backgrounds,fit,shapes.geometric,matrix}
\usepackage{subcaption}

\usepackage{pgfplots}
\usepackage{pgfplotstable}
\pgfplotsset{compat=1.18}
\usepgfplotslibrary{groupplots}

\definecolor{cL3}{HTML}{1f5fbf}   
\definecolor{cL4}{HTML}{d1622b}   
\definecolor{cGuar}{HTML}{888888} 
\definecolor{cS1}{HTML}{9ecae1}   
\definecolor{cS2}{HTML}{6baed6}   
\definecolor{cS3}{HTML}{3182bd}   
\definecolor{cS4}{HTML}{08519c}   
\definecolor{cOnline}{HTML}{2c7fb8}
\definecolor{cOracle}{HTML}{d95f0e}

\newcommand{\panelW}{2.62cm}
\newcommand{\panelH}{1.62cm}   

\pgfplotsset{
  expbase/.style={
    width=\panelW, height=\panelH,
    scale only axis=true,
    font=\scriptsize,
    tick label style={font=\tiny},
    label style={font=\scriptsize},
    title style={font=\scriptsize, yshift=-1.0ex},
    every axis title/.style={above, at={(0.5,1.0)}},
    axis line style={line width=0.4pt},
    tick style={line width=0.3pt, color=black!50},
    major tick length=1.6pt,
    enlarge x limits=0.08,
    line width=0.7pt,
    mark size=1.0pt,
    legend style={font=\tiny, draw=none, fill=none, inner sep=1pt,
                  row sep=-2pt, /tikz/every even column/.append style={column sep=3pt}},
  },
  exprow/.style={expbase},
}

\makeatletter
\newtheoremstyle{acmplaintight}%
  {3\p@\@plus2\p@\@minus2\p@}
  {3\p@\@plus2\p@\@minus2\p@}
  {\@acmplainbodyfont}{\@acmplainindent}{\@acmplainheadfont}{.}{.5em}%
  {\thmname{#1}\thmnumber{ #2}\thmnote{ {\@acmplainnotefont(#3)}}}%
\newtheoremstyle{acmdefinitiontight}%
  {3\p@\@plus2\p@\@minus2\p@}
  {3\p@\@plus2\p@\@minus2\p@}
  {\@acmdefinitionbodyfont}{\@acmdefinitionindent}{\@acmdefinitionheadfont}{.}{.5em}%
  {\thmname{#1}\thmnumber{ #2}\thmnote{ {\@acmdefinitionnotefont(#3)}}}%
\makeatother

\theoremstyle{acmplaintight}
\newtheorem{theorem}{Theorem}
\newtheorem{lemma}{Lemma}
\newtheorem{corollary}{Corollary}
\newtheorem{proposition}{Proposition}
\newtheorem{definition}{Definition}
\newtheorem{remark}{Remark}
\newtheorem{problem}{Problem}

\newtheorem{property}{Property}

\theoremstyle{acmdefinitiontight}
\newtheorem{example}{Example}
\newtheorem{observation}{Observation}
\theoremstyle{acmplaintight}

\newcommand{\mR}{\ensuremath{\mathcal{R}}\xspace} \newcommand{\mA}{\ensuremath{\mathcal{A}}\xspace}   

\newcommand{\HIN}{\ensuremath{G=(V,E)}\xspace}
\newcommand{\mpath}{\ensuremath{\mathcal{P}}\xspace}          
\newcommand{\plen}{\ensuremath{i}\xspace}                      
\newcommand{\pfam}{\ensuremath{\mathcal{V}}\xspace}             
\newcommand{\inst}[1]{\ensuremath{\mathcal{F}(#1)}\xspace}     
\newcommand{\dens}{\ensuremath{\rho}\xspace}
\newcommand{\densopt}{\ensuremath{\rho^{*}}\xspace}
\newcommand{\densoptM}[1][\mathbf{m}]{\ensuremath{\rho^{*}_{#1}}\xspace}
\newcommand{\irm}{\ensuremath{\mathbf{m}}\xspace}              
\newcommand{\irmopt}{\ensuremath{\mathbf{m}^{*}}\xspace}

\newcommand{\irmall}{\ensuremath{\mathbb{M}}\xspace}
\newcommand{\peelout}{\ensuremath{\widehat{\gamma}}\xspace}    
\newcommand{\dleft}{\ensuremath{L}\xspace}
\newcommand{\dright}{\ensuremath{R}\xspace}
\newcommand{\cnt}[2]{\ensuremath{P(#1,#2)}\xspace}             
\newcommand{\eps}{\ensuremath{\varepsilon}\xspace}
\newcommand{\OO}{\ensuremath{\mathcal{O}}\xspace}

\makeatletter
\g@addto@macro\normalsize{%
  \setlength\abovedisplayskip{3pt plus 1pt minus 1pt}%
  \setlength\belowdisplayskip{3pt plus 1pt minus 1pt}%
  \setlength\abovedisplayshortskip{0pt plus 1pt}%
  \setlength\belowdisplayshortskip{2pt plus 1pt minus 1pt}%
}
\makeatother
\setlist[itemize]{partopsep=0pt}

\makeatletter
\renewenvironment{proof}[1][\proofname]{\par
  \pushQED{\qed}%
  \normalfont \topsep3\p@\@plus2\p@\relax
  \trivlist
  \item[\@proofindent\hskip\labelsep
        {\@proofnamefont #1\@addpunct{.}}]\ignorespaces
}{%
  \popQED\endtrivlist\@endpefalse
}
\makeatother

\AtBeginDocument{%
  }

\begin{document}

\everymath{\scriptstyle}
\everydisplay\expandafter{\the\everydisplay \scriptsize\scriptstyle}
\tikzset{every picture/.append style={execute at begin picture={\everymath{}}}}

\title[From Enumeration to Covering: Near-Optimal DPSS over Large HINs]{From Enumeration to Covering: Near-Optimal Densest $\mathcal{P}$-Partite Subgraph Search over Large Heterogeneous Information Networks}
\subtitle{Technical Report}

\author{\mbox{Lu Chen, Chengfei Liu, Rui Zhou}}
\affiliation{%
  \institution{Swinburne University of Technology}
  \city{Melbourne}
  \country{Australia}
}
\email{\{luchen, cliu, rzhou\}@swin.edu.au}

\author{Jiajie Xu}
\affiliation{%
  \institution{Soochow University}
  \city{Suzhou}
  \country{China}
}
\email{xujj@suda.edu.cn}

\author{Jianxin Li}
\affiliation{%
  \institution{Edith Cowan University}
  \city{Perth}
  \country{Australia}
}
\email{jianxin.li@ecu.edu.au}

\begin{abstract}
Given a heterogeneous information network (HIN) and a query meta-path \mpath
of length $\plen$, the densest \mpath-partite subgraph problem finds the
subgraph, spanning the $\plen$ typed layers of \mpath, that maximizes a
parameter-free density: the number of meta-path instances over the geometric
mean of the layer sizes.
It has applications across bibliographic, e-commerce, and
biomedical networks. The state-of-the-art approximation linearizes the
geometric-mean objective by fixing per-layer weights, but solves one
subproblem for every feasible weight set, of which there are
$\OO((n/\plen)^{\plen})$, and on each achieves only a $1/\plen$
approximation. We show that neither the exhaustive enumeration nor the loose guarantee
is necessary. First, we replace enumeration by covering: polylogarithmically
many representative weight sets, localized further by a data-dependent bound,
cover all feasible ones while losing only a tunable factor $1+\eta$ in density.
Second, we cast each fixed-weight subproblem as a weighted supermodular
densest-subgraph instance and solve it near-optimally, lifting the overall
guarantee to $(1-\delta)/(1+\eta)$. To our knowledge, this is the first
near-optimal density approximation beyond the bipartite ($\plen{=}2$) case, and
it yields a PTAS for every fixed $\plen$. Algorithmically, our solver is an adaptive
peeling scheme that never materializes the meta-path instances, whose number
can exceed the graph size by orders of magnitude. Incumbent-driven pruning further discards representative weight sets
before their subproblems are solved. Experiments on five real HINs show that our algorithms
achieve substantial speedups over enumeration-based baselines, and
can further certify the near-optimality of the returned subgraph.


\end{abstract}


\maketitle

\pagestyle{plain}
\begingroup\small\noindent\raggedright\textbf{PVLDB Artifact Availability:}\\
The source code, data, and/or other artifacts have been made available at \url{\vldbavailabilityurl}.
\endgroup

\section{Introduction}\label{sec:intro}

\noindent\textbf{Heterogeneous information networks.}
A heterogeneous information network (HIN)~\cite{shi2017survey} carries a schema $T=(\mA,\mR)$ of object
types $\mA$ and relation types $\mR$: a graph \HIN{} with a vertex-type map
$\phi:V\to\mA$ and edge-type map $\psi:E\to\mR$, $|\mA|+|\mR|>2$. HINs model
bibliographic, e-commerce, knowledge-graph, and biomedical data such as
Hetionet~\cite{himmelstein2017hetionet}, which links compounds, genes, diseases,
anatomies, and others
(Figure~\ref{fig:hetionet-schema}). The schema lets one query at the type level through a
\emph{meta-path}~\cite{sun2011pathsim}: a path
$A_1\xrightarrow{R_1}A_2\xrightarrow{R_2}\cdots\xrightarrow{R_{\plen-1}}A_{\plen}$ on
the schema $T$ that alternates object types $A\in\mA$ and relation types
$R\in\mR$, composing them into a single typed relation between its endpoints. When
the connecting relations are understood we abbreviate it by its \emph{sequence of
vertex types} $\mpath=\langle A_1,\dots,A_{\plen}\rangle$; e.g.\ over Hetionet
$\mpath=\langle\textsf{C},\textsf{G},\textsf{D}\rangle$ (compound--gene--disease)
links a drug to diseases via the genes it binds, of interest to drug repurposing.

\noindent\textbf{Cohesive \mpath-partite subgraphs.}
A \mpath-family is a tuple of vertex sets, one per type in \mpath; the subgraph it
induces, the \mpath-partite graph, is typically large and sparse, while its meaningful information is hidden in the
dense regions. Mining cohesion across multiple vertex sets has driven a long line of
models, from multipartite cliques~\cite{grunert2002kpartite} and cross-graph
quasi-cliques~\cite{pei2005crossgraph} to multilayer cores and densest
subgraphs~\cite{galimberti2017multilayer, chen2023densest}, the relational
community~\cite{jian2020effective} and influential community~\cite{zhou2023influential}.
Among these models, the densest \mpath-partite subgraph is especially appealing: a
simple, parameter-free query surfacing cohesive structure across \emph{every} type in
\mpath, with no size, degree, or pattern threshold to tune. It stands to the
pattern-based $(k,\mathrm P)$-core~\cite{fang2020effective} as the densest subgraph
does to the $k$-core.

\noindent\textbf{Densest \mpath-partite subgraph search (\textsf{DPSS}).}
Given a meta-path \mpath, \textsf{DPSS}~\cite{chen2023densest} seeks a cohesive and
semantically meaningful subgraph within the corresponding
\mpath-partite graph.
Its density is defined as the ratio between the number of
\mpath-instances contained in the subgraph and the geometric mean of the
sizes of its vertex sets across all vertex types. Mathematically, for a
\mpath-family $\pfam=(V_1,\ldots,V_{\plen})$ of $\plen$
typed layers,
\begin{equation}\label{eq:density}
\dens(\pfam)
=
\frac{|\inst{\pfam}|}
{\left(\prod_{j=1}^{\plen}|V_j|\right)^{1/\plen}},
\end{equation}
where $\inst{\pfam}$ is the set of \mpath-instances induced by \pfam,
and \textsf{DPSS} returns the family of maximum density $\densopt$.
The geometric-mean normalization is what makes the objective
parameter-free: it balances all $\plen$ layers simultaneously, so that a
high-density family must concentrate many meta-path instances among
\emph{relatively few} vertices of \emph{every} type.


\definecolor{cmpColor}{RGB}{36,123,160}   
\definecolor{genColor}{RGB}{217,138,8}    
\definecolor{disColor}{RGB}{197,58,50}    
\definecolor{auxColor}{RGB}{120,120,120}  
\definecolor{netColor}{RGB}{70,130,90}    
\definecolor{trapColor}{RGB}{120,80,160}  

\providecommand{\eidcn}[2]{\tikz[baseline=-0.55ex]{\node[circle,draw=#1,
  line width=0.4pt,inner sep=0.2pt,minimum size=2.7mm,font=\tiny]{#2};}}
\providecommand{\eid}[1]{\eidcn{auxColor!90}{#1}}
\providecommand{\eidr}[1]{\tikz[baseline=-0.55ex]{\node[rounded corners=1pt,
  draw=auxColor!90,line width=0.4pt,inner sep=1.2pt,font=\tiny]{#1};}}
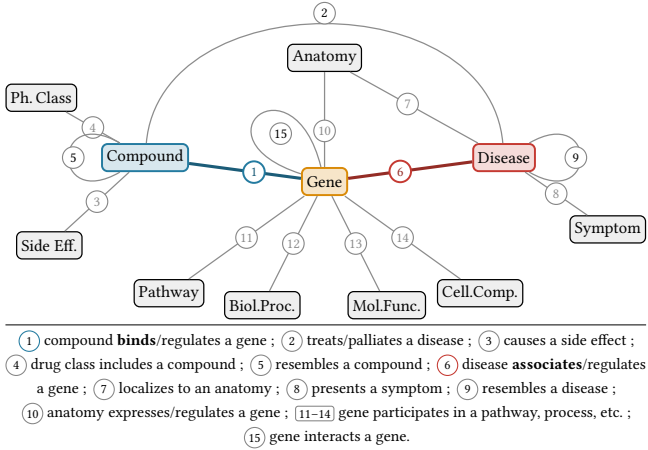
\begin{figure}[t]
\centering
\resizebox{\columnwidth}{!}{%
\begin{tikzpicture}[
  metanode/.style={draw, rounded corners=2pt, inner sep=1.7pt, minimum height=3.4mm,
                   fill=auxColor!12, line width=0.45pt, font=\scriptsize},
  cmp/.style={metanode, fill=cmpColor!18, draw=cmpColor, line width=0.7pt},
  gen/.style={metanode, fill=genColor!22, draw=genColor, line width=0.7pt},
  dis/.style={metanode, fill=disColor!18, draw=disColor, line width=0.7pt},
  rel/.style={auxColor!85, line width=0.45pt},
  mp/.style ={line width=1.1pt},
  eid/.style={circle, draw=auxColor!90, fill=white, inner sep=0pt,
              minimum size=2.7mm, font=\tiny, line width=0.4pt}
]
\node[gen] (G)  at (0,0)        {Gene};
\node[cmp] (C)  at (-2.25,0.3)  {Compound};
\node[dis] (D)  at (2.25,0.3)   {Disease};
\node[metanode] (A)  at (0,1.55)     {Anatomy};
\node[metanode] (PC) at (-3.55,1.05) {Ph.\,Class};
\node[metanode] (SE) at (-3.45,-0.8) {Side Eff.};
\node[metanode] (S)  at (3.55,-0.6)  {Symptom};
\node[metanode] (PW) at (-1.95,-1.4) {Pathway};
\node[metanode] (BP) at (-0.78,-1.55){Biol.Proc.};
\node[metanode] (MF) at (0.78,-1.55) {Mol.Func.};
\node[metanode] (CC) at (1.95,-1.4)  {Cell.Comp.};
\draw[mp, cmpColor!75!black] (C) -- (G) node[eid,draw=cmpColor,line width=0.7pt,pos=0.58]{1};
\draw[mp, disColor!75!black] (G) -- (D) node[eid,draw=disColor,line width=0.7pt,pos=0.42]{6};
\draw[rel] (C)  to[bend left=85,looseness=1.15] (D);
\node[eid] at (0,2.11) {2};
\draw[rel] (C)  -- (SE)                         node[eid,pos=0.5]{3};
\draw[rel] (PC) -- (C)                          node[eid,pos=0.5]{4};
\draw[rel] (C)  to[out=150,in=210,looseness=7] (C);
\node[eid] at (-3.15,0.3) {5};
\draw[rel] (D)  -- (A)                          node[eid,pos=0.55]{7};
\draw[rel] (D)  -- (S)                          node[eid,pos=0.5]{8};
\draw[rel] (D)  to[out=30,in=-30,looseness=7] (D);
\node[eid] at (3.15,0.3) {9};
\draw[rel] (A)  -- (G)                          node[eid,pos=0.62]{10};
\draw[rel] (G)  -- (PW)                         node[eid,pos=0.5]{11};
\draw[rel] (G)  -- (BP)                         node[eid,pos=0.5]{12};
\draw[rel] (G)  -- (MF)                         node[eid,pos=0.5]{13};
\draw[rel] (G)  -- (CC)                         node[eid,pos=0.5]{14};
\draw[rel] (G)  to[out=162,in=104,looseness=14] (G);
\node[eid] at (-0.55,0.6) {15};
\end{tikzpicture}}
\par\smallskip\hrule height0.3pt\smallskip
{\scriptsize\setlength{\parindent}{0pt}%
\eidcn{cmpColor}{1}~compound \textbf{binds}/regulates a gene 
;\,
\eid{2}~treats/palliates a disease 
;\,
\eid{3}~causes a side effect 
;\,
\eid{4}~drug class includes a compound 
;\,
\eid{5}~resembles a compound 
;\,
\eidcn{disColor}{6}~disease \textbf{associates}/regulates a gene 
;\,
\eid{7}~localizes to an anatomy 
;\,
\eid{8}~presents a symptom 
;\,
\eid{9}~resembles a disease 
;\,
\eid{10}~anatomy expresses/regulates a gene 
;\,
\eidr{11--14}~gene participates in a pathway, process, etc. 
;\,
\eid{15}~gene interacts a gene. 
\par}
\caption{Schema of the Hetionet biomedical HIN
}
\label{fig:hetionet-schema}
\end{figure}



\begin{figure}[t]
\centering
\resizebox{0.74\columnwidth}{!}{%
\begin{tikzpicture}[
  v/.style={circle, draw, minimum size=3.6mm, inner sep=0pt, line width=0.55pt, font=\tiny},
  cmpv/.style={v, fill=cmpColor!18, draw=cmpColor},
  genv/.style={v, fill=genColor!22, draw=genColor},
  disv/.style={v, fill=disColor!18, draw=disColor},
  eopt/.style={line width=0.7pt, black},
  etrap/.style={line width=0.55pt, auxColor, densely dashed}
]
\node[font=\tiny\bfseries, cmpColor!70!black] at (0,2.65)   {Compounds};
\node[font=\tiny\bfseries, genColor!60!black] at (2.6,2.65) {Genes};
\node[font=\tiny\bfseries, disColor!70!black] at (5.2,2.65) {Diseases};
\node[cmpv] (c2) at (0, 2.1)  {$c_2$};
\node[cmpv] (c3) at (0, 1.5)  {$c_3$};
\node[cmpv] (c7) at (0, 0.9)  {$c_7$};
\node[cmpv] (c4) at (0, 0.3)  {$c_4$};
\node[cmpv] (c6) at (0,-0.3)  {$c_6$};
\node[cmpv] (c8) at (0,-0.9)  {$c_8$};
\node[cmpv] (c5) at (0,-1.5)  {$c_5$};
\node[cmpv] (c1) at (0,-2.1)  {$c_1$};
\node[genv] (g2) at (2.6, 1.5)  {$g_2$};
\node[genv] (g4) at (2.6, 0.5)  {$g_4$};
\node[genv] (g1) at (2.6,-0.5)  {$g_1$};
\node[genv] (g3) at (2.6,-1.5)  {$g_3$};
\node[disv] (d1) at (5.2, 1.5)  {$d_1$};
\node[disv] (d2) at (5.2, 0.7)  {$d_2$};
\node[disv] (d3) at (5.2,-0.5)  {$d_3$};
\node[disv] (d4) at (5.2,-1.5)  {$d_4$};
\draw[eopt] (c3)--(g2); \draw[eopt] (c4)--(g2);
\draw[eopt] (c2)--(g4); \draw[eopt] (c3)--(g4); \draw[eopt] (c7)--(g4);
\draw[eopt] (g2)--(d1); \draw[eopt] (g2)--(d2); \draw[eopt] (g4)--(d2);
\draw[etrap] (c1)--(g1); \draw[etrap] (c5)--(g1); \draw[etrap] (c8)--(g1); \draw[etrap] (g1)--(d3);
\draw[etrap] (c4)--(g3); \draw[etrap] (c6)--(g3); \draw[etrap] (c8)--(g3); \draw[etrap] (g3)--(d4);
\end{tikzpicture}}
\par\smallskip\hrule height0.3pt\smallskip
{\scriptsize\setlength{\parindent}{0pt}%
\textcolor{cmpColor!75!black}{\textbf{Compounds}}: $c_1$ masitinib, $c_2$ avapritinib,
$c_3$ imatinib, $c_4$ dasatinib, $c_5$ crenolanib, $c_6$ axitinib, $c_7$ ripretinib,
$c_8$ sunitinib.\;
\textcolor{genColor!60!black}{\textbf{Genes}}: $g_1$ \textsf{PDGFRA}, $g_2$ \textsf{ABL1},
$g_3$ \textsf{KDR}, $g_4$ \textsf{KIT}.\;
\textcolor{disColor!75!black}{\textbf{Diseases}}: $d_1$ chronic myeloid leukemia (CML),
$d_2$ gastrointestinal stromal tumor (GIST), $d_3$ hypereosinophilic syndrome,
$d_4$ renal cell carcinoma.\par}
\caption{A $\langle\textsf{C},\textsf{G},\textsf{D}\rangle$-partite graph ($\plen{=}3$)
}
\label{fig:running-instance}
\end{figure}
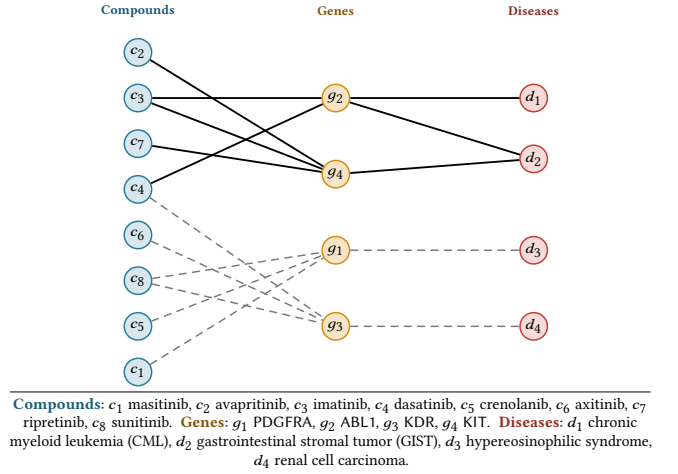

\noindent\textbf{Applications.}
Each densest family is a cohesive group spanning \emph{every} type on the meta-path. Prior
work motivates \textsf{DPSS} with such groups across HINs~\cite{chen2023densest}, exposing
fraud rings over user--device--merchant networks, finding research communities over
author--paper--venue networks, etc. We broaden the case to biomedical drug repurposing.
Take $\mpath=\langle\textsf{C},\textsf{G},\textsf{D}\rangle$ (compound--gene--disease) over
a biomedical network such as Hetionet: a densest family is a set of compounds, genes, and
diseases all densely interlinked across the three layers. In
Figure~\ref{fig:running-instance}, the solid $(4,2,2)$ block $\mathcal V^{*}$ is four kinase
inhibitors acting through two genes (\textsf{ABL1}, \textsf{KIT}) on two cancers, CML and
GIST. This block \emph{suggests} drug-repurposing candidates: because \emph{imatinib}, developed
for CML via \textsf{ABL1}, shares the dense subgraph with GIST through \textsf{KIT}, the result
identifies it as a candidate, a repurposing that in fact holds~\cite{dagher2002imatinib}. Pairwise, edge-based methods are likely to miss it: the query places no edge between a
compound and a disease; such a connection therefore tends to surface only through the \emph{joint}
density across all three layers.


\noindent\textbf{State of the art.}
\textsf{DPSS} is polynomial-time solvable for fixed $\plen$, but only inefficiently, through
costly flow computations; approximation is the more promising direction. The
state-of-the-art approximation~\cite{chen2023densest} fixes a per-layer weight set with
product one, the \emph{iRM set}, to linearize the coupling geometric-mean denominator,
turning \textsf{DPSS} into a weighted densest-subgraph subproblem. Two weaknesses remain: the optimal weight set is unknown, forcing a
search over $\OO((n/\plen)^{\plen})$ of them, and each is solved only to the loose
$1/\plen$ factor of peeling. Both are already overcome at $\plen=2$, the directed/bipartite
densest subgraph~\cite{charikar2000greedy,khuller2009dense}: there the weight set collapses
to a single scalar ratio that is geometrically scanned, and each fixed-ratio subproblem is
supermodular and solved near-optimally. This relies on one-dimensionality and does not
survive $\plen>2$.

\noindent\textbf{Open problems and challenges.}
The obstacle is twofold: for $\plen>2$ the weight sets form a
$(\plen-1)$-dimensional hyperplane with \emph{no} total order to sweep, and the numerator
turns from the \emph{edge} count into the \emph{meta-path instance} count $|\inst{S}|$,
implicit and far larger than the graph size. Three challenges result.
\emph{(C1)} The enumeration of weight sets is $\OO(\prod_j|V(A_j)|)$, exponential in
$\plen$; can approximation avoid it? \emph{(C2)} The $1/\plen$ worst-case factor of
peeling is often near-optimal in practice, yet never \emph{certified} to be; can a
fixed-weight subproblem, with its per-type-weighted denominator, be solved with a
certified near-optimal guarantee? \emph{(C3)} The
materialized instances number $\OO(\prod_j|V_j|)$; can peeling avoid them?


\noindent\textbf{Our approach.}
We answer all three.
\emph{(C1)} We replace enumeration with covering. In logarithmic coordinates we flatten the feasibility $\prod_j m_j=1$ to the hyperplane $\sum_j\ln m_j=0$, where the feasible weights fill a bounded $(\plen-1)$-dimensional polytope; a single $\eta$-net covers it. We space that net at the coarse $\sqrt\eta$ pitch of the \emph{second-order net}, which the curvature of the AM--GM gap justifies, keeping only $\OO((\log n/\sqrt\eta)^{\plen-1})$ representatives at a $1+\eta$ loss, $\eta^{-(\plen-1)/2}$ times fewer than the naive $\eta$ pitch, and we shrink $\log n$ to $\log\Delta$, where $\Delta$ is the largest per-vertex instance count, with a data-dependent bound.
\emph{(C2)} We cast each fixed-weight subproblem as a weighted supermodular densest-subgraph problem and solve it natively to a certified $(1-\delta)$ under its \emph{per-type-weighted} denominator, the part the supermodular machinery, built for the cardinality ratio (the $i{=}2$ case), leaves out. That $(1-\delta)$ lives on the \emph{weighted} objective, not the density; we prove it crosses over, since the AM--GM factor between them \emph{vanishes at a family's own weights}, so at the optimum's weight our bound attains $\densopt$ itself.
\emph{(C3)} For meta-paths, we factor each per-vertex incidence as $\dleft(v)\dright(v)$, so an adaptive peel reads it in $\OO(|E|)$ and achieves the same $(1-\delta)$ \emph{without ever materializing} the $\OO(\prod_j|V_j|)$ instances.
Together we get a $\frac{1-\delta}{1+\eta}$ approximation to $\densopt$: a PTAS for fixed
$\plen$ and, to our knowledge, the first near-optimal guarantee past $\plen=2$.
Finally, we add incumbent-driven \emph{pruning} that removes representatives before we ever solve them.
 


\noindent\textbf{Contributions.}
In summary, this paper makes the following contributions.
\begin{itemize}[leftmargin=1.2em,itemsep=2pt,topsep=0pt]
\item {\sloppy \emph{Covering, not enumeration} (C1): instead of the $\OO((n/\plen)^{\plen})$
weight sets prior work enumerates, an $\eta$-net of only
$\OO((\min\{\log n,\log\Delta\}/\sqrt{\eta})^{\plen-1})$ representatives, polylogarithmic in
$n$ and $\Delta$-adaptive, covers all of them at a $1+\eta$ loss
(\S\ref{sec:geometry}\footnote{Throughout, \S\,$n$ abbreviates Section~$n$.}).\par}
\item \emph{Certified near-optimal subproblems} (C2): each fixed-weight step is a supermodular
densest-subgraph instance, which we solve to a certified $(1-\delta)$ under its
\emph{per-type-weighted} denominator, the part the supermodular machinery leaves out; with the
cover, this yields a $\frac{1-\delta}{1+\eta}$ PTAS for $\densopt$ at fixed $\plen$
(\S\ref{sec:supermod}).
\item \emph{Instance-free peeling} (C3): for path meta-paths, an $\dleft(v)\dright(v)$
factorization runs the peel in $\OO(|E|)$, \emph{without materializing} the
$\OO(\prod_j|V_j|)$ $\mpath$-instances (\S\ref{sec:implicit}).
\item \emph{Pruning}: an incumbent-driven test further removes representatives and reduces subproblem sizes (\S\ref{sec:pruning}).
\item Experiments on real HINs show consistent gains in efficiency
 over existing methods with certified effectiveness (\S\ref{sec:experiments}).
\end{itemize}

\section{Problem Formulation}\label{sec:prelim}

\noindent
In \S\ref{sec:intro}, we have defined the heterogeneous information
network \HIN, the query meta-path \mpath\ of length $\plen=|\mpath|$, the
\mpath-family $\pfam=(V_1,\dots,V_\plen)$ and its induced \mpath-partite
graph, the instance set \inst{\pfam}, and the density \dens. We do not
restate them. 
We introduce the remaining notation and formally state the problem; throughout, $n=|V|$. Table~\ref{tab:notation}
summarizes all symbols.


\noindent\textbf{Typed vertices and positions.}
A
meta-path may \emph{repeat} a vertex type ($A_j=A_k$ for $j\neq k$), as in the gene--disease--gene
path $\langle\textsf{G},\textsf{D},\textsf{G}\rangle$, and its repeated positions are filled
\emph{independently}: a family may pick different sets $V_j\neq V_k$, of different sizes, for
them. The type $\phi(v)$ then no longer identifies a vertex's role: when a type recurs,
knowing $\phi(v)$ leaves it ambiguous which position $v$ occupies. 
We therefore index by position and work over the \emph{position-expanded} ground set
$
U=\biguplus_{j=1}^{\plen} V(A_j),
$
one copy of $V(A_j)$ per position, so every $v\in U$ carries a definite position
$\phi'(v)\in[\plen]$ (of required type $A_{\phi'(v)}$). A \mpath-family is then a subset of
$U$, equivalently a tuple $\pfam=(V_1,\dots,V_\plen)$ with $V_j\subseteq V(A_j)$, and a
vertex whose type recurs appears once per position, chosen freely at each. Every per-position
quantity is thereby unambiguous. When all types are distinct, $\phi'$ is fixed by $\phi$ and
$U\cong V$. 

\noindent\textbf{Induced subgraph and per-vertex instances.}
A \mpath-family $\pfam=(V_1,\dots,V_\plen)$ induces the subgraph
$G(\pfam)$ on $V_1\cup\cdots\cup V_\plen$, and a \mpath-instance of
$G(\pfam)$ is a sequence $\pi=(v_1,\dots,v_\plen)$ with $v_j\in V_j$ whose
consecutive vertices are joined by the relation that \mpath\ prescribes; we write
$v\in\pi$ when $v$ is one of its vertices.
By convention positions are matched independently, so a vertex may fill more than
one role (i.e.\ $v_j=v_k$ is allowed when $A_j=A_k$) and every position-respecting
sequence counts as a distinct instance.
The algorithms that will be discussed operate not on \inst{\pfam} directly but on its
per-vertex incidence. As such, for a vertex $v$,  we write
$
\cnt{v}{G(\pfam)}
=
\bigl|\{\,\pi\in\inst{\pfam}:v\in\pi\,\}\bigr|
$
for the number of \mpath-instances of $G(\pfam)$ passing through $v$.

\begin{table}[t]
\centering
\small
\setlength{\tabcolsep}{5pt}
\caption{Summary of notation.}
\label{tab:notation}
\begin{tabular}{@{}ll@{}}
\toprule
Symbol & Meaning \\
\midrule
\HIN & input HIN (typed graph); $n=|V|$ \\
$\mpath,\ \plen$ & meta-path; length $\plen=|\mpath|$ \\
$A_j,\ V(A_j)$ & position-$j$ type; its vertex set \\
$U,\ \phi'(v)$ & position-expanded $\biguplus_j V(A_j)$; position of $v$ \\
$\phi(v)$ & type of vertex $v$ \\
$\pfam$ & \mpath-family $(V_1,\dots,V_\plen)$, $V_j\subseteq V(A_j)$ \\
$G(\pfam)$ & subgraph induced by $\pfam$ \\
$\inst{\pfam}$ & \mpath-instances in $G(\pfam)$ \\
$\cnt{v}{G(\pfam)}$ & \# \mpath-instances through $v$ \\
$g(\pfam)$ & $\big(\prod_{j}|V_j|\big)^{1/\plen}$ \\
$\mathbf m,\ m_j$ & iRM set; layer-$j$ weight $g(\pfam)/|V_j|$ \\
$\Pi_n$ & polytope of feasible iRM sets (\S\ref{sec:geometry}) \\
$\dens,\ \densopt$ & density; optimum density \\
\bottomrule
\end{tabular}
\end{table}

\begin{example}
\label{ex:instances}
Consider $\pfam^{*}=(\{c_2,c_3,c_4,c_7\},\{g_2,g_4\},\{d_1,d_2\})$, the solid
edges of Figure~\ref{fig:running-instance}: $g_2$ binds $c_3,c_4$ and reaches
$d_1,d_2$, while $g_4$ binds $c_2,c_3,c_7$ and reaches $d_2$. Its
\mpath-instances (the compound--gene--disease paths) number 7.
The per-vertex incidences $\cnt{v}{G(\pfam^{*})}$ are $4,3$ at $g_2,g_4$; $1,3,2,1$
at $c_2,c_3,c_4,c_7$; and $2,5$ at $d_1,d_2$.
\end{example}

\begin{problem}[Densest \mpath-Partite Subgraph Search, \textsf{DPSS}]
\label{prob:dpss}
Given a HIN \HIN\ and a meta-path \mpath, return a \mpath-family
\[
\pfam^{*}=\arg\max_{\pfam}\ \dens(\pfam),
\qquad
\densopt=\dens(\pfam^{*}),
\]
where $\dens(\pfam)=|\inst{\pfam}|/g(\pfam)$ is the density of Eq.~\eqref{eq:density}.
\end{problem}

\noindent\textbf{Scope.} \textsf{DPSS} is defined for every $\plen\ge2$, but the case $\plen=2$
is directed/bipartite densest subgraph, already solved near-optimally by existing
methods~\cite{ma2022convexdds,liang2026racd}. This paper targets the open general regime
$\plen>2$. 


\begin{example}
\label{ex:density}
The family $\pfam^{*}$ of Example~\ref{ex:instances} has density
$
\dens(\pfam^{*})=\frac{7}{\sqrt[3]{4\cdot2\cdot2}}
\approx 2.78,
$
and it is the densest \mpath-family in
Figure~\ref{fig:running-instance}, so $\densopt\approx2.78$. Enlarging it with the
gray decoy ones ($g_1,g_3$ and their compounds $c_1,c_5,c_6,c_8$ and diseases
$d_3,d_4$) gives the whole graph, of shape $(8,4,4)$ with $13$ instances, yet its
density is only
$13/\sqrt[3]{128}\approx2.58<\densopt$. 
\end{example}

\noindent\textbf{The iRM ($\plen$th-root-multiplication) set~\cite{chen2023densest}}.
For a \mpath-family \pfam, the per-layer weights are
$\mathbf m=(m_1,\dots,m_\plen)$ with
\[
m_j=\frac{g(\pfam)}{|V_j|},\qquad j\in[\plen],
\]
rescaling each layer by its size. Fixing it turns \textsf{DPSS} into a tractable
weighted subproblem (to be discussed in \S\ref{sec:revisit}). 
\begin{property}
\label{prop:irm-product}
For every \mpath-family \pfam, the induced iRM set satisfies
$\prod_{j=1}^{\plen} m_j = 1$.
\end{property}

\noindent\textbf{iRM set extension.}
We extend the notion beyond family-induced weights: \emph{any} vector $\mathbf m$ with
$\prod_j m_j=1$ is a \emph{feasible} iRM set, and one actually induced by a \mpath-family is
\emph{realizable}. 

\section{State of the Art}\label{sec:revisit}

\noindent
We revisit the state of the art for \textsf{DPSS} and its $\plen=2$ special case,
isolating the gaps to fill.

\noindent\textbf{The DPSS approximation.}
The geometric-mean density $\dens(\pfam)=|\inst{\pfam}|/g(\pfam)$ resists direct
optimization because its denominator $g=(\prod_j|V_j|)^{1/\plen}$ couples all \plen\ layers
multiplicatively. Fixing a target iRM-set $\irm=(m_1,\dots,m_\plen)$ breaks the coupling and linearizes the denominator. Over the
position-expanded ground set $U$ of \S\ref{sec:prelim} (with $S_j=S\cap V(A_j)$ for
$S\subseteq U$ and position map $\phi'$), fixing $\irm$ replaces the geometric mean by the
\emph{weighted size}
\begin{equation}\label{eq:wsize}
w_{\irm}(S)=\sum_{v\in S}m_{\phi'(v)}=\sum_j m_j|S_j|,
\end{equation}
a weighted vertex count \emph{linear} in $S$, equal to $\plen\,g(S)$ exactly when $\irm$ is
$S$'s own iRM-set.
The fixed-weight subproblem is then to maximize the \emph{auxiliary objective}
\[
h_{\irm}(S)=\frac{f(S)}{w_{\irm}(S)},\qquad f(S):=|\inst{S}|,
\]
the instance count over a linear denominator, a \emph{weighted densest-subgraph} problem,
which \cite{chen2023densest} solves by \emph{reweighted peeling}: it treats each vertex's incidence
$\cnt{v}{G(\pfam)}$ as its value and its weight $m_{\phi'(v)}$ as its cost, repeatedly deletes a
vertex of least value-per-cost $\cnt{v}{G(\pfam)}/m_{\phi'(v)}$, and keeps the densest family
seen. A charging argument, each optimum instance charged to one of its \plen\
endpoints, bounds the loss for $\irm$. 

\begin{theorem}[\cite{chen2023densest}]
\label{thm:chen-peeling}
For a fixed iRM-set $\irm$, reweighted peeling is a
$1/\plen$-approximation to the optimum density $\densoptM$ over all
\mpath-families conforming to $\irm$ (i.e., whose own induced iRM-set equals the fixed $\irm$).
\end{theorem}

\noindent
Since $\irmopt$ is unknown, \cite{chen2023densest} peels at \emph{every} realizable iRM-set
$\irm\in\irmall$ ($\irmall$ the set of realizable iRM-sets) and returns the best, a global
$1/\plen$-approximation, at the
cost of $\Theta(|\irmall|)=\OO\!\big((n/\plen)^{\plen}\big)$ peelings.
Besides, iRM-sets can only be pruned using flow networks. As such, pure peeling-based approximation is even slower than the exact algorithms \cite{chen2023densest}.

\noindent\textbf{The $\plen=2$ advantage.}
At $\plen=2$, directed/bipartite densest subgraph, the iRM-set collapses to a
single scalar \emph{ratio} ($m_2=1/m_1$, $m_1\in[n^{-1/2},n^{1/2}]$). A
\emph{multiplicative} $(1{+}\eps)$ approximation of the ratio suffices, so geometric
spacing spans the whole interval with only $\OO(\eps^{-1}\log n)$ guesses, the
``$\log n$ trick''~\cite{ma2022convexdds} that replaces Charikar's $\Theta(n^2)$ ratio
sweep~\cite{charikar2000greedy,khuller2009dense}. Each fixed-ratio subproblem is a
supermodular program solved to $(1{-}\eps)$ by accelerated coordinate
descent~\cite{nguyen2024acdm,liang2026racd}. Both the configuration search and the
per-subproblem solver are thus already near-optimal.

\textit{Why this does not lift.}
Both pillars presuppose \emph{one-dimensionality}. The search exploits the \emph{total order} of the
ratio interval: a $(1{+}\eps)$ change in the ratio moves the density only by $(1{+}\OO(\eps))$, so
one geometric scan suffices and the order lets core decomposition and divide-and-conquer discard
sub-intervals that cannot hold the optimum. Once $\plen-1>1$, Property~\ref{prop:irm-product}
makes the feasible iRM-sets a $(\plen-1)$-dimensional hyperplane with \emph{no} total order, and
hence no axis to sweep and no sub-interval to prune. The optimum should therefore be
\emph{covered} in the hyperplane, not searched in a line. The
subproblem solver fails for two reasons: at $\plen=2$ the numerator is the \emph{edge} count, a
supermodular function the existing solvers exploit, whereas at $\plen>2$ it is the \emph{meta-path
instance} count $|\inst{S}|$, whose supermodularity is not a priori clear; and even granting it,
the per-type iRM weights turn the denominator into a general $\plen$-type weighted size
$\sum_j m_j|S_j|$, beyond the two-sided $m_1|S_1|{+}\tfrac1{m_1}|S_2|$ those solvers handle.


\section{Solution Overview}\label{sec:overview}

\noindent
Our approach builds on a simple \emph{cover-and-conquer} principle that
decouples the two difficulties of \textsf{DPSS}: searching the iRM
configuration space and solving each fixed-weight subproblem.
Rather than enumerating every iRM-set, we partition the
iRM polytope $\Pi_n$ into a small number of \emph{non-overlapping} cells,
solve a \emph{single} representative iRM-set in each, and return the
densest family found.

Algorithm~\ref{alg:overview} is parameterized by two components, each supplied by a later
section: a \emph{partition} of $\Pi_n$ into disjoint cells $C_1\sqcup\cdots\sqcup C_K$
(\S\ref{sec:geometry}), and a \emph{cell solver} $\textsc{Solve}(G,\mpath,\irm)$ returning a
\mpath-family for a fixed iRM-set $\irm$ (\S\ref{sec:supermod}). It never searches for
$\irmopt$; it solves one representative per cell, on the premise that the cell containing
$\irmopt$ already yields a near-optimal family. Three questions decide whether the premise holds.

\begin{algorithm}[t]
\small
\caption{\textsc{Cover-and-Conquer} framework}
\label{alg:overview}
\KwIn{HIN $G$; meta-path $\mpath$}
\KwOut{A \mpath-family $\widehat{\pfam}$}
Partition the iRM polytope $\Pi_n$ into non-overlapping cells
$C_1,\dots,C_K$, $\widehat{\pfam}\leftarrow\emptyset$\;
\ForEach{cell $C_k$}{
  pick a representative iRM-set $\widehat{\irm}_k\in C_k$\;
  $\pfam_k\leftarrow\textsc{Solve}(G,\mpath,\widehat{\irm}_k)$\;
  \lIf{$\dens(\pfam_k)>\dens(\widehat{\pfam})$}{$\widehat{\pfam}\leftarrow\pfam_k$}
}
\Return{$\widehat{\pfam}$}\;
\end{algorithm}

\noindent\emph{(Q1)~Does a small partition bound the density error?}
The representative $\widehat{\irm}_k$ need not be induced by a family, so the partition must
be both \emph{few} and \emph{faithful}: solving at $\widehat{\irm}_k$ rather than at the
optimum's own weights must cost little density. \S\ref{sec:geometry} delivers both: an
$\eta$-net covers $\Pi_n$ with only $\OO((\log n/\sqrt\eta)^{\plen-1})$ cells (polylogarithmic
in $n$, against the $\OO((n/\plen)^{\plen})$ enumeration), each within a tunable
\emph{covering factor} $1+\eta$ of the optimum.

\noindent\emph{(Q2)~Can each cell be solved near-optimally?}
At a fixed iRM-set the subproblem is still a hard densest-subgraph problem.
\S\ref{sec:supermod} supplies a near-optimal cell solver returning, for a tunable $\delta$, a
family within a \emph{solver factor} $\beta=1-\delta$ of the best at $\irm$, running without ever
materializing the meta-path instances (\S\ref{sec:implicit}).

\noindent\emph{(Q3)~Do the two approximation errors compose?}
Both fall on each cell and multiply, and the cell holding $\irmopt$ ties them to the global
optimum: its \emph{cell optimum} $\densoptM[C_k]:=\max\{\dens(\pfam):\mathrm{iRM}(\pfam)\in
C_k\}$ ($\mathrm{iRM}(\pfam)$ the iRM-set $\pfam$ induces) equals $\densopt$ and the realizable $\irmopt$ makes $\pfam^{*}$ a competitor
there, so that cell's \textsc{Solve} certifies density $\ge\tfrac{\beta}{1+\eta}\densopt$, and
reporting the densest family over all cells improves it. The result is a
$\tfrac{\beta}{1+\eta}=\tfrac{1-\delta}{1+\eta}$ approximation, near-optimal
for every fixed $\plen$.

\section{From Enumeration to Covering}\label{sec:geometry}

\noindent
The cover promised in \S\ref{sec:revisit} comes down to one fact: approximation does not need the
\emph{exact} optimal iRM-set $\irmopt$, only one \emph{close} to it, so that the
exponential search for $\irmopt$ collapses to solving at a few fixed \emph{representatives}.

\noindent\textbf{Closeness in \emph{weights} yields closeness in \emph{density}}. Recall from
\S\ref{sec:revisit} that fixing the weights replaces the coupling geometric mean in the density
by the linear \emph{weighted size} $w_{\irm}(S)=\sum_j m_j|S_j|$ (Eq.~\ref{eq:wsize}), equal to
$\plen g$ at a family's own iRM weights. 
Solving instead at a nearby
representative $\irm'$ scores the optimum $\pfam^{*}$ at the \emph{wrong} weights: each
$m^{*}_j$ becomes $m'_j$, so (using $|V^{*}_j|=g(\pfam^{*})/m^{*}_j$) its weighted size grows from
$w_{\irmopt}(\pfam^{*})=\plen\,g(\pfam^{*})$ to $w_{\irm'}(\pfam^{*})=\sum_j m'_j|V^{*}_j|$ (by AM--GM). Since the
solver reads density as the \emph{weighted density} $\rho_{\irm}(S)=\plen|\inst{S}|/w_{\irm}(S)$
(exact at a family's own weights), the density it assigns
$\pfam^{*}$ at $\irm'$ is deflated by a \emph{loss factor} $L$:
\[
\frac{\rho_{\irm'}(\pfam^{*})}{\densopt}
=\frac{w_{\irmopt}(\pfam^{*})}{w_{\irm'}(\pfam^{*})}
=\frac{\plen\,g(\pfam^{*})}{\sum_j m'_j|V^{*}_j|}
=(\tfrac1\plen \sum_{j} m'_j/m^{*}_j)^{-1}
=\frac1L .
\]
Here $L:=\tfrac1\plen\sum_{j} m'_j/m^{*}_j\ge1$, with equality iff $\irm'=\irmopt$, by AM--GM. The optimum thus attains only
$\densopt/L$ at the representative; the closer $\irm'$ to $\irmopt$, the closer $L$ to $1$, so a small
set of representatives covers $\irmopt$ within a controlled density loss. Since $L$ is
governed by the per-axis ratios $m'_j/m^{*}_j$, covering \emph{every} possible $\irmopt$ at
once means distributing the representatives \emph{multiplicatively}, so that each feasible
iRM-set lies within a small ratio of one on every axis.

\noindent\textbf{The change of coordinates.}
A multiplicative grid is awkward in native coordinates.
Feasibility $\prod_{j} m_j=1$ is a \emph{curved} surface that no axis-aligned grid meets, and
the closeness we need is the \emph{multiplicative} ratio $m'_j/m_j$, so a uniformly spaced
grid in $\mathbf m$ is needlessly fine near the origin and uselessly coarse far from it.
Passing to \emph{logarithms} $x_j=\ln m_j$ cures both: feasibility collapses to the
\emph{additive} $\sum_{j} x_j=0$, a flat $(\plen-1)$-dimensional hyperplane, and a
multiplicative ratio becomes an additive displacement, so a \emph{single uniform lattice}
in $\mathbf x$ spaces the representatives multiplicatively, exactly as the loss factor
demands.

\noindent\textbf{Roadmap.}
\S\ref{ssec:polytope} bounds the targets in a flat $(\plen-1)$-polytope $\Pi_n$;
\S\ref{ssec:net} covers it with an $\eta$-net at a coarse $\sqrt\eta$ pitch ($1+\eta$ loss);
\S\ref{ssec:reduction} turns weight-closeness into a density
$\tfrac{\beta}{1+\eta}$ guarantee; and \S\ref{ssec:deltabox} shrinks the net to the sub-area
where the optimum can occupy.

\subsection{The Bounded Log-iRM Polytope}
\label{ssec:polytope}

Write an iRM-set as $\mathbf m=(m_1,\ldots,m_{\plen})$ and pass to
logarithmic coordinates $x_j=\ln m_j$, $\mathbf x=(x_1,\ldots,x_{\plen})$.
By the feasibility condition of
Property~\ref{prop:irm-product}, $\prod_{j=1}^{\plen} m_j=1$  linearizes
to $\sum_{j=1}^{\plen} x_j=0$.

\begin{lemma}\label{lem:hyperplane}
In logarithmic coordinates the feasible iRM-sets form the
$(\plen-1)$-dimensional affine subspace (a hyperplane)
\[
\mathcal H=\Big\{\mathbf x\in\mathbb R^{\plen}:
\textstyle\sum_{j=1}^{\plen} x_j=0\Big\}.
\]
\end{lemma}

\begin{proof}
Under the bijection $\mathbf x\mapsto e^{\mathbf x}$, $\prod_j m_j=1$ is equivalent
to $\sum_j x_j=0$, one linear constraint of codimension one.
\end{proof}

So the $\plen$ coordinates carry only $\plen-1$ free parameters. For
$\langle\textsf{C},\textsf{G},\textsf{D}\rangle$ ($\plen=3$),
Figure~\ref{fig:net-variants}(a) draws the resulting $2$-dimensional plane, which is a finite \emph{hexagon}. We settle its boundary below.

Each $m_j$ is located between the smallest and largest values, and under logarithms, the bound is a hypercube. The lemma below makes this precise.

\begin{lemma}\label{lem:box}
Every feasible iRM-set satisfies\footnote{The uniform cap $|V_k|\le n$ is for
readability; with per-type sizes $n_k=|V(A_k)|$ the same computation yields the
per-axis bounds $n_j^{-(\plen-1)/\plen}\le m_j\le\big(\prod_{k\neq j}n_k\big)^{1/\plen}$.}
\[
n^{-(\plen-1)/\plen}\ \le\ m_j\ \le\ n^{(\plen-1)/\plen}
\qquad\forall\, j\in[\plen].
\]
\end{lemma}

\begin{proof}
$m_j=\big(\prod_{k\neq j}|V_k|/|V_j|^{\plen-1}\big)^{1/\plen}$; with $1\le|V_k|\le n$
the numerator is $\le n^{\plen-1}$ and the denominator $\ge1$, giving the upper
bound, and symmetrically the lower bound.
\end{proof}

Taking logarithms, every feasible point lies in the hypercube
\[
B_n=\Big[-\tfrac{\plen-1}{\plen}\ln n,\ \tfrac{\plen-1}{\plen}\ln n\Big]^{\plen},
\]
The bound treats each layer as free in $[1,n]$; the actual budget
$\sum_k|V_k|\le n$ forbids those extremal profiles (a coordinate reaches
$n^{(\plen-1)/\plen}$ only if $\plen-1$ parts are each of size $n$), so $B_n$
strictly over-encloses the feasible region, which only tightens the cover.
Intersecting the hypercube with the hyperplane yields the
feasible region
$
\Pi_n=B_n\cap\mathcal H,
$
a bounded $(\plen-1)$-dimensional polytope 
that contains every feasible iRM-set. That settles boundedness.

At $\plen=3$, $\Pi_n$ is the hypercube $B_n$ sliced perpendicular to its main diagonal
$(1,1,1)$, a regular hexagon; in general it is a $(\plen-1)$-polytope with up to $2\plen$ facets.

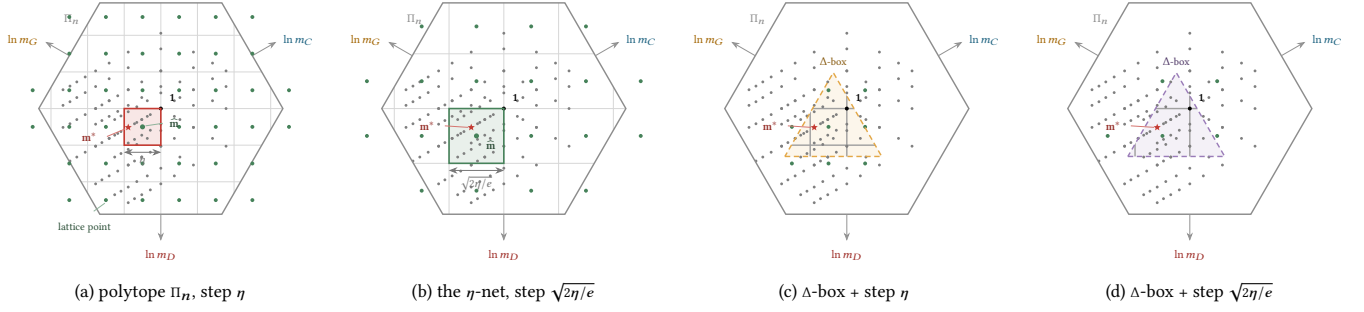
\begin{figure*}[t]
\centering
\begin{minipage}[t]{0.235\textwidth}\centering\vspace{0pt}
\resizebox{\linewidth}{!}{%
\begin{tikzpicture}[
  >={Stealth[length=1.6mm]},
  gridline/.style={auxColor!30, line width=0.3pt},
  Mdot/.style={auxColor!95, fill=auxColor!95},
  netpt/.style={netColor, fill=netColor},
  opt/.style={disColor, fill=disColor}
]
\def\s{1.3082}
\def\Hex{(2.4,0)--(1.2,2.078)--(-1.2,2.078)--(-2.4,0)--(-1.2,-2.078)--(1.2,-2.078)--cycle}
\fill[disColor!12] (-0.72,-0.72) rectangle (0,0);
\begin{scope}
  \clip \Hex;
  \foreach \i in {-4,...,4}{\draw[gridline] (\i*0.72,-2.4)--(\i*0.72,2.4);}
  \foreach \j in {-4,...,4}{\draw[gridline] (-2.6,\j*0.72)--(2.6,\j*0.72);}
  \foreach \a in {1,...,8}{\foreach \b in {1,...,4}{\foreach \cc in {1,...,4}{
     \pgfmathsetmacro{\ux}{ln(\b/\a)/sqrt(2)*\s}
     \pgfmathsetmacro{\vy}{ln((\cc*\cc)/(\a*\b))/sqrt(6)*\s}
     \fill[Mdot] (\ux,\vy) circle(0.8pt);
  }}}
\end{scope}
\draw[line width=0.8pt, auxColor!85] \Hex;
\node[font=\scriptsize, auxColor!85] at (-1.75,1.75) {$\Pi_n$};
\draw[->, auxColor!80, line width=0.5pt] (1.8,1.039)--(2.27,1.31);
\node[font=\scriptsize, cmpColor!80!black, anchor=west] at (2.3,1.35) {$\ln m_C$};
\draw[->, auxColor!80, line width=0.5pt] (-1.8,1.039)--(-2.27,1.31);
\node[font=\scriptsize, genColor!75!black, anchor=east] at (-2.3,1.35) {$\ln m_G$};
\draw[->, auxColor!80, line width=0.5pt] (0,-2.078)--(0,-2.6);
\node[font=\scriptsize, disColor!80!black, anchor=north] at (0,-2.66) {$\ln m_D$};
\fill[black] (0,0) circle(1.1pt);
\node[font=\scriptsize, anchor=south west] at (0.06,0.04) {$\mathbf 1$};
\foreach \p in {(-2.52,-0.36),(-2.52,0.36),(-1.8,-1.8),(-1.8,-1.08),(-1.8,-0.36),(-1.8,0.36),(-1.8,1.08),(-1.8,1.8),(-1.08,-1.8),(-1.08,-1.08),(-1.08,-0.36),(-1.08,0.36),(-1.08,1.08),(-1.08,1.8),(-0.36,-1.8),(-0.36,-1.08),(-0.36,-0.36),(-0.36,0.36),(-0.36,1.08),(-0.36,1.8),(0.36,-1.8),(0.36,-1.08),(0.36,-0.36),(0.36,0.36),(0.36,1.08),(0.36,1.8),(1.08,-1.8),(1.08,-1.08),(1.08,-0.36),(1.08,0.36),(1.08,1.08),(1.08,1.8),(1.8,-1.8),(1.8,-1.08),(1.8,-0.36),(1.8,0.36),(1.8,1.08),(1.8,1.8),(2.52,-0.36),(2.52,0.36)}{\fill[netpt] \p circle(1.1pt);}
\draw[disColor, line width=0.9pt] (-0.72,-0.72) rectangle (0,0);
\draw[<->, auxColor, line width=0.4pt] (-0.72,-0.86)--(0,-0.86);
\node[font=\scriptsize, auxColor, anchor=north] at (-0.36,-0.86) {$\eta$};
\fill[netpt] (-0.36,-0.36) circle(1.4pt);
\node[font=\scriptsize, netColor!55!black, anchor=west] at (0.06,-0.28) {$\widehat{\mathbf m}$};
\draw[netColor!60, line width=0.3pt] (-0.30,-0.34)--(0.04,-0.29);
\node[font=\scriptsize, netColor!60!black, anchor=north] at (-1.5,-2.14) {lattice point};
\draw[netColor!60, line width=0.3pt] (-1.32,-2.02)--(-1.1,-1.86);
\node[opt,star,star points=5,star point ratio=2.3,inner sep=0pt,
      minimum size=4.2pt] at (-0.641,-0.370) {};
\node[font=\scriptsize, disColor!75!black, anchor=east] at (-1.15,-0.55) {$\mathbf m^{*}$};
\draw[disColor!70, line width=0.3pt] (-1.06,-0.55)--(-0.70,-0.40);
\end{tikzpicture}}\\[1pt]
{\footnotesize (a) polytope $\Pi_n$, step $\eta$}
\end{minipage}\hfill
\begin{minipage}[t]{0.235\textwidth}\centering\vspace{0pt}
\resizebox{\linewidth}{!}{%
\begin{tikzpicture}[
  >={Stealth[length=1.6mm]},
  gridline/.style={auxColor!30, line width=0.3pt},
  Mdot/.style={auxColor!95, fill=auxColor!95},
  netpt/.style={netColor, fill=netColor},
  opt/.style={disColor, fill=disColor}
]
\def\Hex{(2.4,0)--(1.2,2.078)--(-1.2,2.078)--(-2.4,0)--(-1.2,-2.078)--(1.2,-2.078)--cycle}
\fill[netColor!12] (-1.077,-1.077) rectangle (0,0);
\begin{scope}
  \clip \Hex;
  \draw[gridline] (-4.308,-2.6)--(-4.308,2.6);
  \draw[gridline] (-2.6,-4.308)--(2.6,-4.308);
  \draw[gridline] (-3.231,-2.6)--(-3.231,2.6);
  \draw[gridline] (-2.6,-3.231)--(2.6,-3.231);
  \draw[gridline] (-2.154,-2.6)--(-2.154,2.6);
  \draw[gridline] (-2.6,-2.154)--(2.6,-2.154);
  \draw[gridline] (-1.077,-2.6)--(-1.077,2.6);
  \draw[gridline] (-2.6,-1.077)--(2.6,-1.077);
  \draw[gridline] (0.000,-2.6)--(0.000,2.6);
  \draw[gridline] (-2.6,0.000)--(2.6,0.000);
  \draw[gridline] (1.077,-2.6)--(1.077,2.6);
  \draw[gridline] (-2.6,1.077)--(2.6,1.077);
  \draw[gridline] (2.154,-2.6)--(2.154,2.6);
  \draw[gridline] (-2.6,2.154)--(2.6,2.154);
  \draw[gridline] (3.231,-2.6)--(3.231,2.6);
  \draw[gridline] (-2.6,3.231)--(2.6,3.231);
  \draw[gridline] (4.308,-2.6)--(4.308,2.6);
  \draw[gridline] (-2.6,4.308)--(2.6,4.308);
  \foreach \a in {1,...,8}{\foreach \b in {1,...,4}{\foreach \cc in {1,...,4}{
     \pgfmathsetmacro{\ux}{ln(\b/\a)/sqrt(2)*1.3082}
     \pgfmathsetmacro{\vy}{ln((\cc*\cc)/(\a*\b))/sqrt(6)*1.3082}
     \fill[Mdot] (\ux,\vy) circle(0.8pt);
  }}}
\end{scope}
\draw[line width=0.8pt, auxColor!85] \Hex;
\node[font=\scriptsize, auxColor!85] at (-1.75,1.75) {$\Pi_n$};
\draw[->, auxColor!80, line width=0.5pt] (1.8,1.039)--(2.27,1.31);
\node[font=\scriptsize, cmpColor!80!black, anchor=west] at (2.3,1.35) {$\ln m_C$};
\draw[->, auxColor!80, line width=0.5pt] (-1.8,1.039)--(-2.27,1.31);
\node[font=\scriptsize, genColor!75!black, anchor=east] at (-2.3,1.35) {$\ln m_G$};
\draw[->, auxColor!80, line width=0.5pt] (0,-2.078)--(0,-2.6);
\node[font=\scriptsize, disColor!80!black, anchor=north] at (0,-2.66) {$\ln m_D$};
\fill[black] (0,0) circle(1.1pt);
\node[font=\scriptsize, anchor=south west] at (0.06,0.04) {$\mathbf 1$};
\foreach \p in {(-2.692,-0.538),(-2.692,0.538),(-1.615,-1.615),(-1.615,-0.538),(-1.615,0.538),(-1.615,1.615),(-0.538,-1.615),(-0.538,-0.538),(-0.538,0.538),(-0.538,1.615),(0.538,-1.615),(0.538,-0.538),(0.538,0.538),(0.538,1.615),(1.615,-1.615),(1.615,-0.538),(1.615,0.538),(1.615,1.615),(2.692,-0.538),(2.692,0.538)}{\fill[netpt] \p circle(1.1pt);}
\draw[netColor, line width=0.9pt] (-1.077,-1.077) rectangle (0,0);
\draw[<->, auxColor, line width=0.4pt] (-1.077,-1.217)--(0,-1.217);
\node[font=\scriptsize, auxColor, anchor=north] at (-0.538,-1.217) {$\sqrt{2\eta/e}$};
\fill[netpt] (-0.538,-0.538) circle(1.4pt);
\node[font=\scriptsize, netColor!55!black, anchor=west] at (-0.46,-0.70) {$\widehat{\mathbf m}$};
\node[opt,star,star points=5,star point ratio=2.3,inner sep=0pt,minimum size=4.2pt] at (-0.641,-0.370) {};
\node[font=\scriptsize, disColor!75!black, anchor=east] at (-1.25,-0.34) {$\mathbf m^{*}$};
\draw[disColor!70, line width=0.3pt] (-1.16,-0.34)--(-0.70,-0.37);
\end{tikzpicture}}\\[1pt]
{\footnotesize (b) the $\eta$-net, step $\sqrt{2\eta/e}$}
\end{minipage}\hfill
\begin{minipage}[t]{0.235\textwidth}\centering\vspace{0pt}
\resizebox{\linewidth}{!}{%
\begin{tikzpicture}[
  >={Stealth[length=1.6mm]},
  gridline/.style={auxColor!70, line width=0.8pt},
  Mdot/.style={auxColor!95, fill=auxColor!95},
  netpt/.style={netColor, fill=netColor},
  opt/.style={disColor, fill=disColor}
]
\def\Hex{(2.4,0)--(1.2,2.078)--(-1.2,2.078)--(-2.4,0)--(-1.2,-2.078)--(1.2,-2.078)--cycle}
\begin{scope}
  \clip \Hex;
  \fill[genColor!8] (-1.218,-0.942)--(0.686,-0.942)--(-0.266,0.707)--cycle;
  \begin{scope}\clip (-1.218,-0.942)--(0.686,-0.942)--(-0.266,0.707)--cycle;
    \draw[gridline] (-4.320,-3)--(-4.320,3);
    \draw[gridline] (-3,-4.320)--(3,-4.320);
    \draw[gridline] (-3.600,-3)--(-3.600,3);
    \draw[gridline] (-3,-3.600)--(3,-3.600);
    \draw[gridline] (-2.880,-3)--(-2.880,3);
    \draw[gridline] (-3,-2.880)--(3,-2.880);
    \draw[gridline] (-2.160,-3)--(-2.160,3);
    \draw[gridline] (-3,-2.160)--(3,-2.160);
    \draw[gridline] (-1.440,-3)--(-1.440,3);
    \draw[gridline] (-3,-1.440)--(3,-1.440);
    \draw[gridline] (-0.720,-3)--(-0.720,3);
    \draw[gridline] (-3,-0.720)--(3,-0.720);
    \draw[gridline] (0.000,-3)--(0.000,3);
    \draw[gridline] (-3,0.000)--(3,0.000);
    \draw[gridline] (0.720,-3)--(0.720,3);
    \draw[gridline] (-3,0.720)--(3,0.720);
    \draw[gridline] (1.440,-3)--(1.440,3);
    \draw[gridline] (-3,1.440)--(3,1.440);
    \draw[gridline] (2.160,-3)--(2.160,3);
    \draw[gridline] (-3,2.160)--(3,2.160);
    \draw[gridline] (2.880,-3)--(2.880,3);
    \draw[gridline] (-3,2.880)--(3,2.880);
    \draw[gridline] (3.600,-3)--(3.600,3);
    \draw[gridline] (-3,3.600)--(3,3.600);
    \draw[gridline] (4.320,-3)--(4.320,3);
    \draw[gridline] (-3,4.320)--(3,4.320);
  \end{scope}
  \foreach \a in {1,...,8}{\foreach \b in {1,...,4}{\foreach \cc in {1,...,4}{
     \pgfmathsetmacro{\ux}{ln(\b/\a)/sqrt(2)*1.3082}
     \pgfmathsetmacro{\vy}{ln((\cc*\cc)/(\a*\b))/sqrt(6)*1.3082}
     \fill[Mdot] (\ux,\vy) circle(0.8pt);
  }}}
\end{scope}
\draw[line width=0.8pt, auxColor!85] \Hex;
\node[font=\scriptsize, auxColor!85] at (-1.75,1.75) {$\Pi_n$};
\draw[->, auxColor!80, line width=0.5pt] (1.8,1.039)--(2.27,1.31);
\node[font=\scriptsize, cmpColor!80!black, anchor=west] at (2.3,1.35) {$\ln m_C$};
\draw[->, auxColor!80, line width=0.5pt] (-1.8,1.039)--(-2.27,1.31);
\node[font=\scriptsize, genColor!75!black, anchor=east] at (-2.3,1.35) {$\ln m_G$};
\draw[->, auxColor!80, line width=0.5pt] (0,-2.078)--(0,-2.6);
\node[font=\scriptsize, disColor!80!black, anchor=north] at (0,-2.66) {$\ln m_D$};
\fill[black] (0,0) circle(1.1pt);
\node[font=\scriptsize, anchor=south west] at (0.06,0.04) {$\mathbf 1$};
\foreach \p in {(-1.08,-1.08),(-1.08,-0.36),(-0.36,-1.08),(-0.36,-0.36),(-0.36,0.36),(0.36,-1.08),(0.36,-0.36),(0.36,0.36)}{\fill[netpt] \p circle(1.1pt);}
\draw[genColor!75, line width=0.8pt, densely dashed] (-1.218,-0.942)--(0.686,-0.942)--(-0.266,0.707)--cycle;
\node[font=\scriptsize, genColor!60!black, anchor=south] at (-0.266,0.747) {$\Delta$-box};
\node[opt,star,star points=5,star point ratio=2.3,inner sep=0pt,minimum size=4.2pt] at (-0.641,-0.370) {};
\node[font=\scriptsize, disColor!75!black, anchor=east] at (-1.25,-0.34) {$\mathbf m^{*}$};
\draw[disColor!70, line width=0.3pt] (-1.16,-0.34)--(-0.70,-0.37);
\end{tikzpicture}}\\[1pt]
{\footnotesize (c) $\Delta$-box $+$ step $\eta$ }
\end{minipage}\hfill
\begin{minipage}[t]{0.235\textwidth}\centering\vspace{0pt}
\resizebox{\linewidth}{!}{%
\begin{tikzpicture}[
  >={Stealth[length=1.6mm]},
  gridline/.style={auxColor!70, line width=0.8pt},
  Mdot/.style={auxColor!95, fill=auxColor!95},
  netpt/.style={netColor, fill=netColor},
  opt/.style={disColor, fill=disColor}
]
\def\Hex{(2.4,0)--(1.2,2.078)--(-1.2,2.078)--(-2.4,0)--(-1.2,-2.078)--(1.2,-2.078)--cycle}
\begin{scope}
  \clip \Hex;
  \fill[trapColor!8] (-1.218,-0.942)--(0.686,-0.942)--(-0.266,0.707)--cycle;
  \begin{scope}\clip (-1.218,-0.942)--(0.686,-0.942)--(-0.266,0.707)--cycle;
    \draw[gridline] (-4.308,-3)--(-4.308,3);
    \draw[gridline] (-3,-4.308)--(3,-4.308);
    \draw[gridline] (-3.231,-3)--(-3.231,3);
    \draw[gridline] (-3,-3.231)--(3,-3.231);
    \draw[gridline] (-2.154,-3)--(-2.154,3);
    \draw[gridline] (-3,-2.154)--(3,-2.154);
    \draw[gridline] (-1.077,-3)--(-1.077,3);
    \draw[gridline] (-3,-1.077)--(3,-1.077);
    \draw[gridline] (0.000,-3)--(0.000,3);
    \draw[gridline] (-3,0.000)--(3,0.000);
    \draw[gridline] (1.077,-3)--(1.077,3);
    \draw[gridline] (-3,1.077)--(3,1.077);
    \draw[gridline] (2.154,-3)--(2.154,3);
    \draw[gridline] (-3,2.154)--(3,2.154);
    \draw[gridline] (3.231,-3)--(3.231,3);
    \draw[gridline] (-3,3.231)--(3,3.231);
    \draw[gridline] (4.308,-3)--(4.308,3);
    \draw[gridline] (-3,4.308)--(3,4.308);
  \end{scope}
  \foreach \a in {1,...,8}{\foreach \b in {1,...,4}{\foreach \cc in {1,...,4}{
     \pgfmathsetmacro{\ux}{ln(\b/\a)/sqrt(2)*1.3082}
     \pgfmathsetmacro{\vy}{ln((\cc*\cc)/(\a*\b))/sqrt(6)*1.3082}
     \fill[Mdot] (\ux,\vy) circle(0.8pt);
  }}}
\end{scope}
\draw[line width=0.8pt, auxColor!85] \Hex;
\node[font=\scriptsize, auxColor!85] at (-1.75,1.75) {$\Pi_n$};
\draw[->, auxColor!80, line width=0.5pt] (1.8,1.039)--(2.27,1.31);
\node[font=\scriptsize, cmpColor!80!black, anchor=west] at (2.3,1.35) {$\ln m_C$};
\draw[->, auxColor!80, line width=0.5pt] (-1.8,1.039)--(-2.27,1.31);
\node[font=\scriptsize, genColor!75!black, anchor=east] at (-2.3,1.35) {$\ln m_G$};
\draw[->, auxColor!80, line width=0.5pt] (0,-2.078)--(0,-2.6);
\node[font=\scriptsize, disColor!80!black, anchor=north] at (0,-2.66) {$\ln m_D$};
\fill[black] (0,0) circle(1.1pt);
\node[font=\scriptsize, anchor=south west] at (0.06,0.04) {$\mathbf 1$};
\foreach \p in {(-1.615,-0.538),(-0.538,-0.538),(-0.538,0.538),(0.538,-0.538),(0.538,0.538)}{\fill[netpt] \p circle(1.1pt);}
\draw[trapColor!75, line width=0.8pt, densely dashed] (-1.218,-0.942)--(0.686,-0.942)--(-0.266,0.707)--cycle;
\node[font=\scriptsize, trapColor!60!black, anchor=south] at (-0.266,0.747) {$\Delta$-box};
\node[opt,star,star points=5,star point ratio=2.3,inner sep=0pt,minimum size=4.2pt] at (-0.641,-0.370) {};
\node[font=\scriptsize, disColor!75!black, anchor=east] at (-1.25,-0.34) {$\mathbf m^{*}$};
\draw[disColor!70, line width=0.3pt] (-1.16,-0.34)--(-0.70,-0.37);
\end{tikzpicture}}\\[1pt]
{\footnotesize (d) $\Delta$-box $+$ step $\sqrt{2\eta/e}$}
\end{minipage}
\caption{
The log-iRM polytope and covering nets
}
\label{fig:net-variants}
\end{figure*}

\subsection{Covering with the $\eta$-Net}
\label{ssec:net}

With the targets boxed into $\Pi_n$, we lay the cover. Two constraints shape the grid:
every grid point must be a \emph{feasible} weight vector ($\prod_j m_j=1$), and we must be
able to \emph{round} an arbitrary target onto it.
A naive axis-aligned grid in the $\plen$ raw log-coordinates meets neither: a
generic point has $\sum_j x_j\neq0$, off the hyperplane $\mathcal H$, so rounding to
it leaves $\mathcal H$ and re-projecting reintroduces error on every coordinate.
Laying the grid \emph{on} $\mathcal H$, as an integer lattice, fixes both at once.

\begin{definition}[Zero-sum log-lattice]\label{def:lattice}
For a step $s>0$ let
\[
\Lambda_s=\Big\{(k_1 s,\ldots,k_{\plen}s):k\in\mathbb Z^{\plen},\ \sum_j k_j=0\Big\}\cap B_n^{+s},
\quad
B_n^{+s}=\big[-\tfrac{\plen-1}{\plen}\ln n-s,\ \tfrac{\plen-1}{\plen}\ln n+s\big]^{\plen},
\]
i.e.\ the lattice points whose log-coordinate $x_j=k_j s$ is an integer multiple of the
step $s$, with the index $k_j=x_j/s\in\mathbb Z$ and $\sum_j k_j=0$ keeping
$\sum_j x_j=0$ (the point on $\mathcal H$). The box is grown by one step per side, the smallest
enlargement whose cells cover $B_n$, so every rounding image of a point in $B_n$ is retained.
Let $\mathcal M^{s}=\{(e^{x_1},\ldots,e^{x_{\plen}}):\mathbf x\in\Lambda_s\}$ be
its image in weight space. Every point of $\mathcal M^{s}$ satisfies
$\prod_j m_j=e^{\sum_j x_j}=1$ \emph{exactly}, with no projection step.
\end{definition}

\begin{lemma}[Zero-sum rounding]\label{lem:round}
For every feasible target $\mathbf x^{*}\in B_n$ with $\sum_j x^{*}_j=0$ there is
$\boldsymbol\lambda\in\Lambda_s$ with $|\lambda_j-x^{*}_j|\le s$ for all $j$.
\end{lemma}

\begin{proof}
Round each coordinate, $k_j=\lfloor x^{*}_j/s+\tfrac12\rfloor$, with error
$e_j=k_js-x^{*}_j\in[-s/2,s/2]$; the defect $r=\sum_j k_j\in\mathbb Z$ has
$|r|\le\plen/2$. If $r\neq0$, adjust $|r|$ coordinates by $\mp1$ (possible since at
least $|r|$ errors share its sign), so each adjusted error lands in $[-s,s)$, the
rest stay in $[-s/2,s/2]$, and now $\sum_j k_j=0$. Every final error has magnitude
at most $s$.
\end{proof}

\noindent\textbf{Choosing the pitch.}
We ask of the lattice one thing: rounding any feasible target to its nearest net point should
keep the density loss to a factor $1+\eta$. A coarse pitch can meet it, as the next
definition and theorem show.

\begin{definition}[$\eta$-net]\label{def:net}
For $\eta\in(0,1]$, set $s(\eta)=\sqrt{2\eta/e}$ and let
$\mathcal M_\eta=\mathcal M^{s(\eta)}$ (the \emph{second-order net}) be the zero-sum
log-lattice of Definition~\ref{def:lattice} at this step.
\end{definition}

\begin{theorem}[A small $\eta$-net]\label{thm:net}
For $\eta\in(0,1]$, the net $\mathcal M_\eta$ has size
$\OO\big((\log n/\sqrt\eta)^{\plen-1}\big)$, and for every feasible $\irmopt$ it
contains a representative $\irm'$ with $\prod_j m'_j=1$ whose loss factor
$L=\tfrac1\plen\sum_j e^{\theta_j}$, $\theta_j=\ln(m'_j/m^{*}_j)$, equals
$1+\Theta(\|\boldsymbol\theta\|^{2})$ and satisfies $L\le1+\eta$.
\end{theorem}

\begin{proof}
$\Lambda_{s(\eta)}$ is $(\plen-1)$-dimensional (the constraint $\sum_j k_j=0$
removes one degree of freedom) and each free coordinate ranges over
$\OO(\log n/s(\eta))=\OO(\log n/\sqrt\eta)$ values (Lemma~\ref{lem:box}), giving the
size. Feasibility is Definition~\ref{def:lattice}; the closeness
$\|\boldsymbol\theta\|_\infty\le s(\eta)$ is Lemma~\ref{lem:round} at
$\mathbf x^{*}=\ln\irmopt$. For the loss, the balance $\sum_j\theta_j=0$ cancels the
linear term,
\[
L=\tfrac1\plen\textstyle\sum_j e^{\theta_j}
=1+\underbrace{\tfrac1\plen\sum_j\theta_j}_{=\,0}
+\tfrac{1}{2\plen}\sum_j\theta_j^{2}+\cdots
=1+\Theta\!\big(\|\boldsymbol\theta\|^{2}\big),
\]
the curvature of the AM--GM gap at equality. For the bound, the elementary inequality
$e^{t}\le1+t+\tfrac{t^{2}}{2}e^{|t|}$ (valid for every real $t$), applied at $t=\theta_j$ and
averaged over $j$, gives $L\le1+\tfrac1\plen\sum_j\theta_j+\tfrac1{2\plen}\sum_j\theta_j^{2}e^{|\theta_j|}$;
the linear sum vanishes by $\sum_j\theta_j=0$, and with $|\theta_j|\le s(\eta)\le1$ each
$e^{|\theta_j|}\le e$, so $L\le1+\tfrac{e}{2}s(\eta)^{2}=1+\eta$.
\end{proof}

\noindent\emph{An intuitive baseline.} The same $1+\eta$ loss is met more straightforwardly by
pinning \emph{each} ratio $m'_j/m^{*}_j$ within $1+\eta$ on its own, the \emph{first-order
net} at pitch $s=\ln(1+\eta)=\Theta(\eta)$. Being quadratically finer, it lays
$\eta^{-(\plen-1)/2}$ times as many cells (Figure~\ref{fig:net-variants}(a) vs.\ (b)); we
keep it only as a baseline.

These net points are the representatives $\widehat{\irm}_k$ that
Algorithm~\ref{alg:overview} solves. Each is the \emph{center} of its cell, the lattice
point every iRM-set there rounds to (Lemma~\ref{lem:round}), so one \textsc{Solve} covers the
whole cell within $1+\eta$. Since $\irmopt$ is unknown, the algorithm solves at
\emph{every} net point and returns the densest family; the cell covering $\irmopt$ is
guaranteed among them, so the global best inherits the $1+\eta$ bound.
On the running instance of Figure~\ref{fig:running-instance}, the $110$ distinct feasible iRM-sets
are covered by 40 net cells, centered at green lattice points (Figure~\ref{fig:net-variants}(a)) and only 20 net cells (Figure~\ref{fig:net-variants}(b)), one \textsc{Solve}
per cell instead of $110$, securing a density loss of
$1+\eta$.

\subsection{How Net Approximation Affects Subproblem Quality}
\label{ssec:reduction}

Theorem~\ref{thm:net} puts a representative $\irm'$ within loss factor $L\le1+\eta$ in the
cell of the unknown optimum $\irmopt$. Running an \emph{existing} fixed-weight solver at
$\irm'$ turns this closeness in \emph{weights} into closeness in \emph{density}: the
\emph{approximate} weights the net hands to the solver cost only the tunable factor
$1+\eta$.


\medskip

\begin{lemma}[Peeling is ratio-robust]\label{lem:perturb}
Let $\pfam^{*}$ be optimal with iRM-set $\irmopt$, and $\irm'$ feasible with loss
factor $\tfrac1\plen\sum_j e^{\theta_j}\le1+\eta$, $\theta_j=\ln(m'_j/m^{*}_j)$
(Theorem~\ref{thm:net}). Peeling with $\irm'$ on any subgraph containing
$G(\pfam^{*})$ yields $\peelout$, the largest density seen during the peel, with
$\densopt\le\plen(1+\eta)\,\peelout$.
\end{lemma}

\begin{proof}
The bound factors into a \emph{weight} step and a \emph{solver} step. Let $S_v$ be the
subgraph present when peeling deletes $v$, and $\mu_j=\max\{\cnt{v}{S_v}:v\in V^{*}_j\}$.

\emph{Weight step.} Charge each instance of $\inst{\pfam^{*}}$ to its first-deleted endpoint;
that endpoint $v$ is deleted while all of $\pfam^{*}$ is still present, so the instance is
counted in $\cnt{v}{S_v}$. Hence
$|\inst{\pfam^{*}}|\le\sum_{v\in\pfam^{*}}\cnt{v}{S_v}\le\sum_j|V^{*}_j|\mu_j$, and dividing by
$g(\pfam^{*})$ with $|V^{*}_j|=g(\pfam^{*})/m^{*}_j$ gives $\densopt\le\sum_j\mu_j/m^{*}_j$.

\emph{Solver step.} The deleted vertex minimizes $\cnt{u}{S_v}/m'_{\phi'(u)}$ over $u\in S_v$,
so this value is at most the $m'$-weighted average
$\plen|\inst{S_v}|/w_{\irm'}(S_v)\le\dens(S_v)\le\peelout$ (AM--GM, $w_{\irm'}\ge\plen g$); thus
$\mu_j\le\peelout\,m'_j$.

Splitting $\mu_j/m^{*}_j=(\mu_j/m'_j)\,e^{\theta_j}$ and using
$\sum_j e^{\theta_j}\le\plen(1+\eta)$ (Theorem~\ref{thm:net}),
\[
\densopt\ \le\ \sum_{j}\frac{\mu_j}{m^{*}_j}
\ \le\ \peelout\sum_{j} e^{\theta_j}
\ \le\ \plen(1+\eta)\,\peelout .
\]
\end{proof}


\begin{theorem}[Peeling over the net]\label{thm:net-peeling}
Running the peeling once at every net point of $\mathcal M_\eta$ and
returning the densest \mpath-family is a $\tfrac{1}{\plen(1+\eta)}$-approximation
to $\densopt$, with $\OO((\log n/\sqrt\eta)^{\plen-1})$ peeling runs.
\end{theorem}

\noindent
The correctness is clear: it is the solver-agnostic reduction above at
$\beta=\tfrac1\plen$. Lemma~\ref{lem:perturb} gives the per-cell guarantee
$\peelout\ge\densopt/(\plen(1+\eta))$ in $\irmopt$'s cell.  Theorem~\ref{thm:net} bounds the number of peeling runs.  


\subsection{Localizing the Hypercube}
\label{ssec:deltabox}

The net size carries a factor $(\log n)^{\plen-1}$ from the worst-case hypercube width
$B_n=\OO(\log n)$ per layer (Lemma~\ref{lem:box}). The \emph{optimum's} weights, though,
cannot roam the whole hypercube: an iRM-weight is not free, and at the optimum the $j$-th
weight is pinned to how instance-rich a typical type-$j$ vertex is. Precisely, $m^{*}_j$
is the \emph{average} number of meta-path instances through a type-$j$ vertex of
$\pfam^{*}$, divided by $\densopt$ (Lemma~\ref{lem:deltabox}). Writing
$\Delta_j=\max_{v\in V(A_j)}P(v,G)$ for the \emph{largest} such per-vertex count, an
average never exceeds its maximum, so $\Delta_j$ caps how far the optimum can lean on
layer $j$.

\begin{lemma}[$\Delta$-box]\label{lem:deltabox}
Assume $\densopt>0$. For the optimal family $\pfam^{*}$,
$m^{*}_j=\bar P_j/\densopt$ where $\bar P_j=|\inst{\pfam^{*}}|/|V^{*}_j|$, and
$1\le\bar P_j\le\Delta_j$. Hence $m^{*}_j\in[1/\densopt,\ \Delta_j/\densopt]$, an
interval of log-width $\ln\Delta_j$.
\end{lemma}

\begin{proof}
$\bar P_j/m^{*}_j=|\inst{\pfam^{*}}|/g(\pfam^{*})=\densopt$ gives the identity. Since
$\sum_{v\in V^{*}_j}P(v,G(\pfam^{*}))=|\inst{\pfam^{*}}|$, the average $\bar P_j\le
\Delta_j$; and $\bar P_j\ge1$, as a type-$j$ vertex in no instance could be deleted
to raise the density (or, if alone, forces $\densopt=0$).
\end{proof}

So the optimum's $j$-th weight is confined to a band of log-width only $\ln\Delta_j$,
\emph{independent of $n$}. Its location depends on the unknown $\densopt$, but feasibility
brackets it regardless. Taking the product of $m^{*}_j=\bar P_j/\densopt$ over $j$ and using
$\prod_j m^{*}_j=1$ gives $\densopt=\big(\prod_j\bar P_j\big)^{1/\plen}$, the geometric mean
of the $\bar P_j$; since each $\bar P_j\in[1,\Delta_j]\subseteq[1,\Delta]$
($\Delta=\max_j\Delta_j$) and a geometric mean lies between its smallest and largest terms,
$\densopt\in[1,\Delta]$. This bracket alone fixes the box, with no estimate of $\densopt$:
in $m^{*}_j=\bar P_j/\densopt$ the numerator is at most $\Delta_j$ and at least $1$, so
$\densopt\ge1$ gives $m^{*}_j\le\Delta_j$ while $\densopt\le\Delta$ gives $m^{*}_j\ge1/\Delta$,
hence $\ln m^{*}_j\in[-\ln\Delta,\ \ln\Delta_j]$ \emph{whatever} $\densopt$ is.
The band thus \emph{provably} contains $\irmopt$: only the bracket $\densopt\in[1,\Delta]$
enters, not an estimate.
This bracket only tightens during the search, where the best density found so far,
$\hat\rho\le\densopt$, sharpens it to $\densopt\in[\hat\rho,\Delta]$ and shrinks the box to
$m^{*}_j\le\Delta_j/\hat\rho$ while still containing $\irmopt$.
The net size is reduced to
$\OO\big((\min\{\log n,\log\Delta\}/\sqrt\eta)^{\plen-1}\big)$.  


\noindent\textbf{$\Delta$ plus net}. 
The net and the $\Delta$-box touch \emph{different} parts of the per-axis base and so
compose without interference. The net (\S\ref{ssec:net}) sets the \emph{pitch}
$\sqrt\eta$, the $\Delta$-box the \emph{range} $\min\{\log n,\log\Delta\}$.
Together they are the cover Algorithm~\ref{alg:overview} sweeps, one fixed-weight \textsc{Solve}
per cell, each within $1+\eta$ of the optimum's
(\S\ref{sec:overview}), and the net we carry into \S\ref{sec:approx}.

\noindent\emph{On the running example.}
Figure~\ref{fig:net-variants}(c) draws the $\Delta$-box: the slab
$m_j\in[1/\densopt,\Delta_j/\densopt]$ (here $\Delta=(3,4,5)$) meets the product-one
plane in a small \emph{triangle}, with $\mathbf m^{*}$ inside. Overlaying the coarse
$\sqrt{2\eta/e}$ pitch on it (Figure~\ref{fig:net-variants}(d)) leaves only 5 cells, a
tiny fraction of the $110$ feasible iRM-sets naive enumeration would visit, yet the
cell of $\mathbf m^{*}$ is among them.

\section{Near-Optimal Approximation for $i>2$}\label{sec:approx}\label{sec:supermod}

\noindent
What remains for a near-optimal algorithm is to push the $\tfrac1\plen$ factor to $1-\delta$.
The road has three steps: the fixed-weight
subproblem is a weighted supermodular densest-subgraph problem in the auxiliary objective
$h_{\irm}$ (\S\ref{ssec:supermodsub}); a native solver makes it near-optimal \emph{in
$h_{\irm}$} (\S\ref{ssec:weightedsolver}); and an AM--GM bridge, tight at the optimum's own
weight, carries that near-optimality \emph{to the density $\dens$} (\S\ref{ssec:bridge}).

\noindent\textbf{Why the near-optimal solver stopped at $\plen=2$.}
The fixed-weight subproblem is a weighted supermodular densest-subgraph problem
(Theorem~\ref{thm:fixed-supermodular}, below), so one might expect the near-optimal supermodular
solvers~\cite{boob2020flowless,chekuri2022supermodular,harb2022faster} to finish the job; the
obstacle is the \emph{per-type weights}. Those solvers target the \emph{cardinality} ratio
$f/|S|$, and the $\plen=2$ methods target only the two-sided
$m_1|S_1|{+}\tfrac1{m_1}|S_2|$~\cite{ma2022convexdds,nguyen2024acdm,liang2026racd}; neither handles a
general $\plen$-type denominator. The possible workaround, \emph{materializing} $m_{\phi'(v)}$ integer copies of each $v$ to feed an
standard \emph{unit-weight} peeler, also fails: the weights are real-valued and spread up to
$n^2$, so exact copies do not exist and rounded ones inflate the graph and corrupt the peeling
order.

\subsection{Supermodular Densest Subgraph}\label{ssec:supermodsub}

Recall the fixed-weight subproblem from \S\ref{sec:revisit}: for an iRM-set $\irm$, maximize the auxiliary objective $h_{\irm}(S)=f(S)/w_{\irm}(S)$ over
$S\subseteq U$, with $f(S)=|\inst{S}|$ and the modular weighted size
$w_{\irm}(S)=\sum_j m_j|S_j|$ (the weighted density of \S\ref{ssec:reduction} is
$\plen\,h_{\irm}$). The denominator $w_{\irm}$ is already modular; the numerator is
supermodular.

\begin{lemma}\label{lem:supermodular}
$f:2^{U}\to\mathbb Z_{\ge 0}$, $f(S)=|\inst{S}|$, is normalized
($f(\emptyset)=0$), monotone, and \emph{supermodular}.
\end{lemma}

\begin{proof}
Normalization and monotonicity are immediate. For supermodularity, fix
$S\subseteq T$, $v\notin T$: every instance in the marginal $f(S\cup\{v\})-f(S)$
passes through $v$ with its other vertices in $S\subseteq T$, so it is also in
$f(T\cup\{v\})-f(T)$.
\end{proof}

A supermodular numerator
over a modular denominator is thus the problem below.

\begin{theorem}\label{thm:fixed-supermodular}
For every feasible iRM-set $\irm$, maximizing $h_{\irm}(S)$ over $S\subseteq U$
is an instance of weighted supermodular densest subgraph (immediate from
Lemma~\ref{lem:supermodular} and the modularity of $w_{\irm}$).
\end{theorem}


\subsection{A Near-Optimal Weighted Solver}
\label{ssec:weightedsolver}

The cardinality-ratio solvers do not apply, and the workaround fails, so we need a native solver. We
adapt \textsc{Super-Greedy++}~\cite{chekuri2022supermodular,boob2020flowless}, which peels the
vertex of least load-augmented degree $\sigma(v)+\deg_S(v)$ and charges its degree back across
passes, in two changes. \emph{(i) Weight the criterion:} delete the vertex of least
\emph{weighted} incidence $\cnt{v}{\cdot}/m_{\phi'(v)}$, putting the per-type weight into the
deletion order. \emph{(ii) Charge and repeat:} carry a per-vertex load across
passes, adding a vertex's incidence when it is deleted, so a vertex peeled too early is spared
next pass and the order self-corrects.

\noindent\textbf{Why repeat: a rising floor meets a descending ceiling.}
Any family peeling returns a \emph{floor}, an achievable density that prior peeling
(Theorem~\ref{thm:chen-peeling}) keeps within a fixed factor of optimal, though a single pass
gives no per-instance certificate. Repeating adds a \emph{ceiling}: the averaged load, peaked
against the per-type weights, upper-bounds the optimum. One deletion order does both, sparing
instance-rich vertices to lift the floor while averaging the charges to lower the ceiling; once
they meet, the family in hand is certified near-optimal. The floor is inherited from prior
peeling, the ceiling and its convergence new, formalized in Proposition~\ref{prop:solver}.

\begin{proposition}[Iterative Peeling]\label{prop:solver}
Peeling
by the load-augmented key $\big(\sigma(v)+\cnt{v}{G(S)}\big)/m_{\phi'(v)}$ with
$\sigma(v)\mathrel{+}=\cnt{v}{G(S)}$ charged at each deletion (ideas (i)--(ii)), reporting the
densest \mpath-family over all prefixes of all passes, returns $S$ with
$h_{\irm}(S)\ge(1-\delta)\,h_{\irm}^{*}$ ($h_{\irm}^{*}=\max_{\emptyset\neq S'\subseteq U}h_{\irm}(S')$) in
$
T=\OO\!\big(\Delta_{\irm}\ln|U|\,/\,(\delta^{2}\,h_{\irm}^{*})\big)\ \text{passes},
$
where $\Delta_{\irm}=\max_{v\in U}\cnt{v}{G}/m_{\phi'(v)}=\max_{j}\Delta_j/m_j$ is the
\emph{weighted width}. 
\end{proposition}

\begin{proof}[Proof sketch]
The argument is self-contained given three lemmas (Lemmas~\ref{lem:greedybase}--\ref{lem:passbase}), each proved in Appendix~\ref{app:prop1}; we
state what each supplies and how they combine.

\emph{(i)~Base membership.} No matter which key chooses the vertex to delete, a pass charges each deleted vertex its
marginal $\cnt{v}{G(S)}=f(v\,|\,S\setminus v):=f(S)-f(S\setminus v)$, so read in reverse deletion order, the charges form a
\emph{greedy base} of the supermodular $f$ (Lemma~\ref{lem:greedybase}, greedy bases and exact
linear minimization). Hence each pass's charge lies in the base polytope
$B(f)=\{x\ge0:x(S)\ge f(S)\ \forall S,\ x(U)=f(U)\}$ (writing $x_v$ for the coordinates of $x\in\mathbb R^{U}$ and $x(S)=\sum_{v\in S}x_v$), and so does the averaged load
$\bar x_t=\sigma_t/t$, $\sigma_t$ the total charge after $t$ passes, by convexity (Lemma~\ref{lem:passbase}, per-pass base membership). 

\emph{(ii)~A computable certificate.} Weak duality on $B(f)$ converts any averaged load into a
bound: for $x\in B(f)$ and $\emptyset\ne S$, $f(S)\le x(S)$ gives
$h_{\irm}(S)\le\max_{v}x_v/m_{\phi'(v)}$, so the gap
$\hat g_t=\max_v\bar x_t(v)/m_{\phi'(v)}-\max_{\text{prefix }S}h_{\irm}(S)$ is nonnegative and
computable, and $\hat g_t\le\delta\,h_{\irm}^{*}$ certifies the reported family
$(1-\delta)$-dense \emph{unconditionally} (Remark~\ref{rem:cert}). 

\emph{(iii)~The gap closes at the stated rate.} The anchor is the weighted quasi-min-max
$h_{\irm}^{*}=\min_{x\in B(f)}\max_v x_v/m_{\phi'(v)}$ (Lemma~\ref{lem:qmm}), the per-type-weighted
counterpart of the cardinality min--max, obtained via the Lov\'asz extension of $f$ and von
Neumann's minimax. Against this anchor, the convergence analysis of
\textsc{Super-Greedy++}~\cite{chekuri2022supermodular} carries over, once loads, marginals and key are read \emph{per unit weight}: the weights enter only through $h_{\irm}^{*}$ and the
width $\Delta_{\irm}$, giving $\min_{t\le T}\hat g_t\le\delta\,h_{\irm}^{*}$ at the stated $T$. On
the $\Delta$-box-localized net the width-to-value ratio collapses to our own parameters,
so the uniform budget $T=\OO(\plen\,\Delta\ln|U|/\delta^{2})$ suffices (Remark~\ref{rem:budget}; $1+\eta\le2$ is absorbed into the constant),
the $1/\delta^{2}$ exponent being proper to peeling-type schemes (Remark~\ref{rem:delta2}). At
$T{=}1$, $\sigma\equiv0$ and the key is $\cnt{v}{G(S)}/m_{\phi'(v)}$, verbatim the reweighted peel of
Theorem~\ref{thm:chen-peeling}; the load carried across passes is the lift to $(1-\delta)$.
\end{proof}

\noindent
By now, we settle the subproblem in the \emph{weighted objective}: for any fixed weight $\irm$, the
solver returns a family within $(1-\delta)$ of $h_{\irm}^{*}$. Our target, however,
is the density $\dens$, not $h_{\irm}$, which will be addressed in 
\S\ref{ssec:bridge}.

\subsection{From $h_{\irm}$ to the Density $\dens$}\label{ssec:bridge}

We now convert near-optimality in $h_{\irm}$ into near-optimality in the density $\dens$. 

\begin{lemma}\label{lem:amgm}
For every $S\subseteq U$ with each $S_j\neq\emptyset$, \
$\dens(S)\ \ge\ \plen\,h_{\irm}(S)$, with equality iff $\irm$ is $S$'s own iRM-set.
\end{lemma}

\begin{proof}
By the AM--GM bound behind the loss factor of \S\ref{sec:geometry},
$w_{\irm}(S)=\sum_j m_j|S_j|\ge\plen(\prod_j|S_j|)^{1/\plen}=\plen\,g(S)$ (using
$\prod_j m_j=1$); dividing into $f(S)$ gives $h_{\irm}(S)\le\dens(S)/\plen$. The bound is tight
iff the $m_j|S_j|$ are all equal, i.e.\ iff $\irm$ is $S$'s own iRM-set.
\end{proof}

\noindent
At the optimum's own iRM-set $\irmopt$ the factor $\plen$ thus cancels,
$\plen\,h_{\irmopt}(\pfam^{*})=\densopt$, so solving there would return
$\dens(S')\ge(1-\delta)\densopt$ directly. Since $\irmopt$ is unknown, the net supplies only a
representative $\irm'$, at the tunable cost $1+\eta$:

\begin{lemma}\label{lem:surrogate}
Let $\irm'$ be the net representative of $\irmopt$'s cell, with displacements
$\theta_j=\ln(m'_j/m^{*}_j)$ and $\tfrac1\plen\sum_j e^{\theta_j}\le1+\eta$
(Theorem~\ref{thm:net}). Then $\plen\,h_{\irm'}(\pfam^{*})\ge\densopt/(1+\eta)$.
\end{lemma}

\begin{proof}
This is the loss-factor identity of \S\ref{sec:geometry}: scoring $\pfam^{*}$ at $\irm'$
deflates its density by exactly $L=\tfrac1\plen\sum_j e^{\theta_j}$, i.e.\
$\plen\,h_{\irm'}(\pfam^{*})=\densopt/L$, and $L\le1+\eta$ by Theorem~\ref{thm:net}.
\end{proof}

\begin{theorem}\label{thm:main}
Running the solver of Proposition~\ref{prop:solver} once at every net point of
the net $\mathcal M_\eta$ (Definition~\ref{def:net}),
localized by the $\Delta$-box (Lemma~\ref{lem:deltabox}), and returning the
densest \mpath-family yields
\[
\dens\ \ge\ \frac{1-\delta}{1+\eta}\,\densopt,
\]
using
$\OO\big((\min\{\log n,\log\Delta\}/\sqrt{\eta})^{\plen-1}\big)$
subproblem solvers, where $\Delta=\max_j\Delta_j$.
\end{theorem}

\begin{proof}
The cell containing $\irmopt$ has a representative $\irm'\in\mathcal M_\eta$
with $\prod_j m'_j=1$ and loss factor $\le1+\eta$
(Theorem~\ref{thm:net}); by Lemma~\ref{lem:deltabox} the localized box
contains that cell. Let $S'$ be the solver's output at $\irm'$.
Lemma~\ref{lem:amgm}'s hypothesis holds at $S'$: by Lemma~\ref{lem:surrogate}
$\plen\,h_{\irm'}^{*}\ge\plen\,h_{\irm'}(\pfam^{*})\ge\densopt/(1+\eta)>0$, so
$h_{\irm'}(S')\ge(1-\delta)h_{\irm'}^{*}>0$, hence $f(S')>0$, and since an instance
occupies every position, every layer of $S'$ is nonempty. Chaining these,
\[
\begin{aligned}
\dens(S')
&\overset{\text{Lem.~\ref{lem:amgm}}}{\ge}\plen\,h_{\irm'}(S')
\overset{(1-\delta)}{\ge}(1-\delta)\,\plen\max_{S}h_{\irm'}(S)\\[2pt]
&\ \ge\ (1-\delta)\,\plen\,h_{\irm'}(\pfam^{*})
\ \overset{\text{Lem.~\ref{lem:surrogate}}}{\ge}\
\frac{1-\delta}{1+\eta}\,\densopt .
\end{aligned}
\]
Reporting the densest family over all net
points only improves the bound, and the number of solvers is that of
Theorem~\ref{thm:net} over the localized box.
\end{proof}

\begin{corollary}[PTAS for fixed $\plen$]\label{cor:ptas}
For every fixed $\plen$ and $\zeta\in(0,1)$, taking $\delta=\eta=\zeta/3$ gives a
$(1-\zeta)$-approximation in time polynomial in $|G|$, a PTAS, not an FPTAS.
\end{corollary}

\begin{proof}
Since $\tfrac{1-\zeta/3}{1+\zeta/3}\ge1-\zeta$, Theorem~\ref{thm:main} gives the
factor. The $1/\zeta$-degree is $\tfrac{\plen+3}{2}$: $(\plen-1)/2$ from the
$\OO\big((\log n/\sqrt{\zeta}\,)^{\plen-1}\big)$ cells and $2$ from the solver's
$\OO(1/\delta^{2})$ passes (Proposition~\ref{prop:solver} with $\delta=\zeta/3$);
the uniform pass budget $T=\OO(\plen\,\Delta\ln|U|/\delta^{2})$
is polynomial in $|G|$ at fixed $\plen$. The degree grows with $\plen$, and the concrete end-to-end runtime
is given in \S\ref{sec:implicit}.
\end{proof}

\noindent\emph{On the running example.}
The optimum is the $(4,2,2)$ module $\pfam^{*}$ of Example~\ref{ex:density}
($\densopt\approx2.78$), interleaved with two decoy fans $g_1,g_3$ of loosely attached
paths. Run the low-cost $\tfrac1\plen$ peel (\S\ref{ssec:reduction}) at \emph{every} net
weight: none isolates the module in a single pass, so the best family it sweeps out is
the whole graph ($\approx2.58$). The supermodular solver with $\delta=0.01$ of this section, maximizing
$h_{\irm}$ to $(1-\delta)$ at the module's own weight, instead returns $\pfam^{*}$
exactly. This $2.78$-against-$2.58$ gap \emph{persists under every net refinement}, a clear
per-cell solver gap. Such interleavings are worst-case; on real data the peel often matches
the solver, but only the solver's $(1-\delta)$ bound \emph{certifies} that no
such gap exists in any cell.

\subsection{Running the Solver Instance-Free}\label{sec:implicit}

\noindent
The previous sections cut the \emph{number} of subproblems; a second bottleneck
hides inside each. Every peeling solver reads the instance set $\inst{\pfam}$, which
for a meta-path can far exceed the graph ($|\inst{\pfam}|=\OO(\prod_j|V_j|)$ against
$|E|$ edges), and materializing it would be wasteful, which we show is unnecessary.

\noindent\textbf{The factorization.}
Fix $\mpath=(A_1,\ldots,A_\plen)$ and $\pfam$. For a vertex
$v\in V(A_j)$ let $\dleft(v)$ be the number of partial instances of the prefix
$(A_1,\ldots,A_j)$ that \emph{end} at $v$, and $\dright(v)$ the number of partial
instances of the suffix $(A_j,\ldots,A_\plen)$ that \emph{start} at $v$, the ways
to reach $v$ from the left, and to continue from $v$ to the right.

\begin{observation}\label{lem:factorization}
 $\forall v$ of $G(\pfam)$, \ $\cnt{v}{G(\pfam)}=\dleft(v)\,\dright(v)$.
\end{observation}


\noindent\textbf{Maintaining counts.}
The two count families are themselves simple path dynamic programs. Writing
$N^{-}(v)$ and $N^{+}(v)$ for the neighbors of $v$ in the preceding and following
layers,
\[
\dleft(v)=\!\!\sum_{u\in N^{-}(v)}\!\!\dleft(u)\ \ (j>1),
\qquad
\dright(v)=\!\!\sum_{u\in N^{+}(v)}\!\!\dright(u)\ \ (j<\plen),
\]
with $\dleft(v)=1$ on the first layer and $\dright(v)=1$ on the last.

\begin{observation}\label{thm:dp}
All counts $\{\dleft(v),\dright(v)\}_{v}$ can be computed by one forward and one
backward sweep over $G(\pfam)$, in $\OO(|E(G(\pfam))|)$ time and space: by the
recurrences above, the forward sweep fills $\dleft$ in increasing $j$ and the backward
sweep fills $\dright$ in decreasing $j$, each reading every edge once.
\end{observation}


\noindent\textbf{Instance-free peeling.}
\textsc{Iterative Peeling} reads the instances only through its
key $(\sigma(v)+\cnt{v}{G(S)})/m_{\phi'(v)}$, and both parts run on $G(\pfam)$ alone: the loads
$\sigma(v)$ are $\OO(|U|)$ scalars updated in place, and the incidence factors as
$\cnt{v}{G(S)}=\dleft(v)\dright(v)$ (Observation~\ref{lem:factorization}), maintained by the two
sweeps of Observation~\ref{thm:dp} and refreshed on deletion by propagating along incident
edges. A pass thus keeps the keys $\big(\sigma(v)+\dleft(v)\dright(v)\big)/m_{\phi'(v)}$ in a
heap and repeatedly extracts the minimum, deletes, and propagates, the \emph{same} execution as
the enumerated solver (since $\dleft(v)\dright(v)=\cnt{v}{G(S)}$), so its $(1-\delta)$ guarantee
carries over with total state $\OO(|E|)$, independent of $|\inst{\pfam}|$. 

\noindent\textbf{End-to-end cost.}
Two quantities, scaling differently in $n$, govern the cost. The \emph{number} of subproblems is
the cell count $K=\OO\big((\min\{\log n,\log\Delta\}/\sqrt\eta)^{\plen-1}\big)$, polylogarithmic in
$n$ and $\Delta$-adaptive (Lemma~\ref{lem:deltabox}), saving over the
$\OO((n/\plen)^{\plen})$ enumeration of prior work. The \emph{cost per subproblem} is polynomial:
on the $\Delta$-box-localized net a cell runs $T=\OO(\plen\,\Delta\ln|U|/\delta^{2})$ passes
(Proposition~\ref{prop:solver}), each costing $\OO(|E|\,\Delta_{\mathrm{prop}}+|U|\log|U|)$ with
$|E|=|E(G(\pfam))|$, $|U|=\sum_j|V(A_j)|$, and $\Delta_{\mathrm{prop}}\le|E|$ the maximum number
of times one pass's count propagation re-reads a single edge. The search therefore
runs in $\OO\big(K\,T\,(|E|\,\Delta_{\mathrm{prop}}+|U|\log|U|)\big)$. 

\noindent\textbf{Discussion.}
The bound on $T$ is worst-case. For the cardinality densest subgraph, Greedy++'s convergence was
conjectured by~\cite{boob2020flowless} and later proved~\cite{chekuri2022supermodular}; faster first-order and coordinate-descent
methods~\cite{harb2022faster,nguyen2024acdm} carry per-edge state that forfeits
instance-freeness. In practice each solver stops far sooner, and the ceiling carries across
representatives $\irm$ (its averaged load is weight-free), trimming the search further
(\S\ref{sec:experiments}); we omit the details due to the limited space.

\section{Vertex and Cell Pruning}\label{sec:pruning}

\noindent
Building on the supermodular reformulation of \S\ref{sec:approx}, we propose a
\emph{per-type} incidence-\emph{floor} vertex pruning that shrinks the graph \emph{before}
any solver. Because the floor scales with each vertex's type weight
$m_{\phi'(v)}$, it is sharper than the single, uniform density threshold of prior
HIN vertex pruning~\cite{chen2023densest}.


Fix a feasible iRM-set $\irm$ with the auxiliary objective $h_{\irm}=f/w_{\irm}$
over the ground set $U$ of \S\ref{sec:approx}, and write
$h_{\irm}^{*}=\max_{\emptyset\neq S\subseteq U}h_{\irm}(S)>0$ for its fixed-weight
optimum. We show the lemma below.

\begin{lemma}[Vertex pruning]\label{lem:coreduce}
With the notation above:
\emph{(i)~Incidence floor.} Every maximizer $S^{*}$ and every $v\in S^{*}$
satisfy $\cnt{v}{G(S^{*})}\ge h_{\irm}^{*}\,m_{\phi'(v)}$, and hence
$\cnt{v}{G}\ge h_{\irm}^{*}\,m_{\phi'(v)}$.
\emph{(ii)~Optimum-preserving peel.} For any $\lambda\le h_{\irm}^{*}$, deleting every vertex
with $\cnt{v}{G}<\lambda\,m_{\phi'(v)}$ removes no vertex of any $h_{\irm}$-optimal
$S^{*}$; iterating the deletion on the shrinking graph against the same $\lambda$
likewise removes none, and leaves $h_{\irm}^{*}$ unchanged.
\end{lemma}

\begin{proof}
Fix a maximizer $S^{*}$; $f(S^{*})=h_{\irm}^{*}w_{\irm}(S^{*})>0$ forces every
layer nonempty, so $w_{\irm}(S^{*})-m_{\phi'(v)}>0$ for every $v\in S^{*}$.
\emph{(i)} For $v\in S^{*}$, deleting $v$ destroys exactly $\cnt{v}{G(S^{*})}$
instances, and optimality of $S^{*}$ over the competitor $S^{*}\!\setminus\!\{v\}$ gives
$\tfrac{f(S^{*})-\cnt{v}{G(S^{*})}}{w_{\irm}(S^{*})-m_{\phi'(v)}}\le h_{\irm}^{*}$;
cross-multiplying yields $\cnt{v}{G(S^{*})}\ge h_{\irm}^{*}m_{\phi'(v)}$, and
$\cnt{v}{G}\ge\cnt{v}{G(S^{*})}$ since $S^{*}\subseteq U$.
\emph{(ii)} By (i) every $v\in S^{*}$ clears the floor $\lambda m_{\phi'(v)}$, so
$S^{*}$ survives each round (with $\cnt{v}{G(U_r)}\ge\cnt{v}{G(S^{*})}$ as
$S^{*}\subseteq U_r$, the ground set after $r$ rounds, by induction) and stays optimal over the shrinking ground set,
leaving $h_{\irm}^{*}$ unchanged.
\end{proof}


The floor of Lemma~\ref{lem:coreduce} involves the unknown $h_{\irm}^{*}$, but
part~(ii) needs only a \emph{lower} bound on it: any $\lambda\le h_{\irm}^{*}$ is a sound
floor, and deleting $\{v:\cnt{v}{G}<\lambda\,m_{\phi'(v)}\}$ removes no vertex of an optimal family.
The running search supplies one, its best auxiliary value $h_{\irm}(S)\le h_{\irm}^{*}$ so far,
which only improves across the run~\cite{chen2023densest}; the peel uses this $\lambda$ in place of
$h_{\irm}^{*}$. When the peel empties a cell's graph, that cell is dropped, and the
bound still holds. 
\section{Experiments}\label{sec:experiments}

\begin{table}[t]
\centering
\caption{Datasets}
\label{tab:datasets}
\scriptsize
\setlength{\tabcolsep}{3.2pt}
\renewcommand{\arraystretch}{1.05}
\begin{tabular}{l rr rr c}
\toprule
Dataset & $|\mathcal A|$ & $|\mathcal R|$ & $|V|$ & $|E|$ & $i$ \\
\midrule
MovieLens & 5    & 6      & $2.7\mathrm{K}$  & $0.10\mathrm{M}$ & 3--4 \\
DBLP      & 4    & 3      & $37.8\mathrm{K}$ & $0.17\mathrm{M}$ & 3--4 \\
Douban    & 6    & 6      & $37.6\mathrm{K}$ & $1.71\mathrm{M}$ & 3--4 \\
DBpedia   & 412  & $661$ & $6.31\mathrm{M}$ & $18.7\mathrm{M}$  & 3--9 \\
Freebase  & $10.1\mathrm{K}$ & $14.8\mathrm{K}$ & $120\mathrm{M}$ & $433\mathrm{M}$ & 3--9 \\
\bottomrule
\end{tabular}

\vspace{9pt}
\definecolor{cNetFirst}{HTML}{d1622b}
\definecolor{cNetSecond}{HTML}{1f5fbf}
\definecolor{cNetLocal}{HTML}{2ca02c}
\definecolor{cCertify}{HTML}{7b3294}
\newcommand{\algmark}[4]{\tikz[baseline=-0.6ex]{%
  \draw[#1,#2,line width=0.7pt] (0,0)--(0.40,0);
  \draw[#1] plot[mark=#3,mark size=1.4pt,mark options={solid,#4}] coordinates {(0.20,0)};}}

\caption{Tested algorithms
}
\label{tab:algs}
\scriptsize
\setlength{\tabcolsep}{4pt}
\renewcommand{\arraystretch}{1.15}
\begin{tabular}{l c l l c}
\toprule
& mark & Net & Per-cell solver & Guarantee \\
\midrule
A1 & ---                                                  & exact (none)        & exact         & $1$                        \\
A2 & \algmark{cNetFirst}{densely dashed}{*}{}             & first-order         & $1/i$ peel    & $\tfrac{1}{i(1+\eta)}$     \\
A3 & \algmark{cNetFirst}{solid}{square*}{}                & first-order         & near-optimal  & $\tfrac{1-\delta}{1+\eta}$\,(conj.) \\
A4 & \algmark{cNetSecond}{densely dashed}{triangle*}{}    & second-order        & $1/i$ peel    & $\tfrac{1}{i(1+\eta)}$     \\
A5 & \algmark{cNetSecond}{solid}{diamond*}{}              & second-order        & near-optimal  & $\tfrac{1-\delta}{1+\eta}$\,(conj.) \\
A6 & \algmark{cNetLocal}{densely dashed}{pentagon*}{}     & local\,+\,second    & $1/i$ peel    & $\tfrac{1}{i(1+\eta)}$     \\
A7 & \algmark{cNetLocal}{solid}{triangle*}{rotate=180}    & local\,+\,second    & near-optimal  & $\tfrac{1-\delta}{1+\eta}$\,(conj.) \\
A8 & \algmark{cCertify}{densely dotted}{star}{}           & local\,+\,second    & near-optimal  & $\tfrac{1-\delta}{1+\eta}$\,(cert.) \\
\bottomrule
\end{tabular}

\end{table}

We now evaluate the effectiveness and efficiency of the proposed techniques.
\subsection{Experimental Setup}
\label{ssec:exp-setup}

\definecolor{cNetFirst}{HTML}{d1622b}
\definecolor{cNetSecond}{HTML}{1f5fbf}
\definecolor{cNetLocal}{HTML}{2ca02c}
\definecolor{cCertify}{HTML}{7b3294}
\newcommand{\expOneCurves}[1]{%
  \addplot[cNetFirst, densely dashed,  mark=*,         mark options={solid}, yshift=1.6pt]  table[x=eta,y=A2]{figures/data/#1};
  \addplot[cNetFirst, solid,   mark=square*,                         yshift=-1.6pt] table[x=eta,y=A3]{figures/data/#1};
  \addplot[cNetSecond, densely dashed, mark=triangle*, mark options={solid}, yshift=-3.2pt] table[x=eta,y=A4]{figures/data/#1};
  \addplot[cNetSecond, solid,  mark=diamond*,                        yshift=1.1pt]  table[x=eta,y=A5]{figures/data/#1};
  \addplot[cNetLocal, densely dashed,  mark=pentagon*, mark options={solid}, yshift=-1.1pt] table[x=eta,y=A6]{figures/data/#1};
  \addplot[cNetLocal, solid,   mark=triangle*, mark options={rotate=180}, yshift=3.2pt]  table[x=eta,y=A7]{figures/data/#1};
  \addplot[cCertify, densely dotted, mark=star, mark options={solid}, yshift=4.8pt]  table[x=eta,y=A7]{figures/data/#1};
}
\newcommand{\legitem}[6]{
  \draw[#2,#3,line width=0.8pt] (#1,0)--(#1+0.50,0);
  \draw[#2] plot[mark=#4,mark size=1.5pt,mark options={solid,#5}] coordinates {(#1+0.25,0)};
  \node[anchor=west,inner sep=1.5pt] at (#1+0.52,0) {#6};}
\newcommand{\expTwoCurves}[1]{%
  \addplot[cNetFirst, solid, mark=square*]                                 table[x=delta,y=A3]{figures/data/#1};
  \addplot[cNetSecond, solid, mark=diamond*,                  yshift=1.5pt] table[x=delta,y=A5]{figures/data/#1};
  \addplot[cNetLocal, solid, mark=triangle*, mark options={rotate=180}, yshift=-1.5pt] table[x=delta,y=A7]{figures/data/#1};
  \addplot[cCertify, densely dotted, mark=star, mark options={solid}, yshift=-3.1pt] table[x=delta,y=A7]{figures/data/#1};
}

\begin{figure*}[t]
\centering
\captionsetup[subfigure]{skip=1pt}
\begin{subfigure}{\textwidth}\centering
\begin{tikzpicture}
\begin{groupplot}[
  expbase,
  width=2.4cm, height=1.7cm,
  group style={group size=5 by 1, horizontal sep=0.95cm},
  xmode=log, x dir=reverse, log basis x=10,
  xtick={0.05,0.1,0.2,0.5,1}, xticklabels={.05,.1,.2,.5,1},
  enlarge x limits=0.06,
  unbounded coords=jump,
  grid=both, grid style={black!10, line width=0.3pt},
  xlabel={$\eta$ (finer $\rightarrow$)},
  ytick={0.99,0.995,1.0}, yticklabels={.99,.995,1.0},
]
\nextgroupplot[title={MovieLens}, ylabel={$\rho/\rho^{*}$},
               ymin=0.989, ymax=1.0012]
  \expOneCurves{exp1_MovieLens.dat}
\nextgroupplot[title={DBLP}, ymin=0.989, ymax=1.0015]
  \expOneCurves{exp1_DBLP.dat}
\nextgroupplot[title={Douban}, ymin=0.984, ymax=1.0015]
  \expOneCurves{exp1_Douban.dat}
\nextgroupplot[title={DBpedia}, ymin=0.997, ymax=1.003,
               ytick={0.998,1.0,1.002}, yticklabels={.998,1,1.002}]
  \expOneCurves{exp1_DBpedia.dat}
  \node[font=\tiny, align=center, text=black!72, fill=white, fill opacity=0.7,
        text opacity=1, inner sep=1pt] at (axis description cs:0.5,0.22)
    {all exact\\($\rho/\rho^{*}{=}1$)};
\nextgroupplot[title={Freebase}, ymin=0.997, ymax=1.003,
               ytick={0.998,1.0,1.002}, yticklabels={.998,1,1.002}]
  \expOneCurves{exp1_Freebase.dat}
  \node[font=\tiny, align=center, text=black!72, fill=white, fill opacity=0.7,
        text opacity=1, inner sep=1pt] at (axis description cs:0.5,0.22)
    {all exact\\($\rho/\rho^{*}{=}1$)};
\end{groupplot}
\coordinate (leg) at ($(group c1r1.north)!0.5!(group c5r1.north)+(0,7mm)$);
\begin{scope}[shift={(leg)}, font=\tiny]
  \legitem{-3.77}{cNetFirst}{densely dashed}{*}{}{A2}
  \legitem{-2.65}{cNetFirst}{solid}{square*}{}{A3}
  \legitem{-1.53}{cNetSecond}{densely dashed}{triangle*}{}{A4}
  \legitem{-0.41}{cNetSecond}{solid}{diamond*}{}{A5}
  \legitem{ 0.71}{cNetLocal}{densely dashed}{pentagon*}{}{A6}
  \legitem{ 1.83}{cNetLocal}{solid}{triangle*}{rotate=180}{A7}
  \legitem{ 2.95}{cCertify}{densely dotted}{star}{}{A8}
\end{scope}
\end{tikzpicture}
\subcaption{Recovered density $\rho/\rho^{*}$ vs.\ net granularity $\eta$.}
\label{fig:exp-eta}
\end{subfigure}\\[-0.4ex]
\begin{subfigure}{\textwidth}\centering
\begin{tikzpicture}
\begin{groupplot}[
  expbase,
  width=2.4cm, height=1.7cm,
  group style={group size=5 by 1, horizontal sep=0.95cm},
  xmode=log, log basis x=10,
  xtick={0.05,0.1,0.25,0.5}, xticklabels={.05,.1,.25,.5},
  enlarge x limits=0.10,
  grid=both, grid style={black!10, line width=0.3pt},
  xlabel={$\delta$ (looser $\rightarrow$)},
  ymin=0.996, ymax=1.001, ytick={0.998,1.0}, yticklabels={.998,1},
]
\nextgroupplot[title={MovieLens}, ylabel={$\rho/\rho^{*}$}]
  \expTwoCurves{exp2_MovieLens.dat}
\nextgroupplot[title={DBLP}]
  \expTwoCurves{exp2_DBLP.dat}
\nextgroupplot[title={Douban}]
  \expTwoCurves{exp2_Douban.dat}
\nextgroupplot[title={DBpedia}]
  \expTwoCurves{exp2_DBpedia.dat}
  \node[font=\tiny, align=center, text=black!72, fill=white, fill opacity=0.7,
        text opacity=1, inner sep=1pt] at (axis description cs:0.5,0.22)
    {all exact\\($\rho/\rho^{*}{=}1$)};
\nextgroupplot[title={Freebase}]
  \expTwoCurves{exp2_Freebase.dat}
  \node[font=\tiny, align=center, text=black!72, fill=white, fill opacity=0.7,
        text opacity=1, inner sep=1pt] at (axis description cs:0.5,0.22)
    {all exact\\($\rho/\rho^{*}{=}1$)};
\end{groupplot}
\end{tikzpicture}
\subcaption{Recovered density $\rho/\rho^{*}$ vs.\ per-subproblem tolerance $\delta$.}
\label{fig:exp-delta}
\end{subfigure}
\caption{Effectiveness evaluation.}
\label{fig:effectiveness}
\end{figure*}
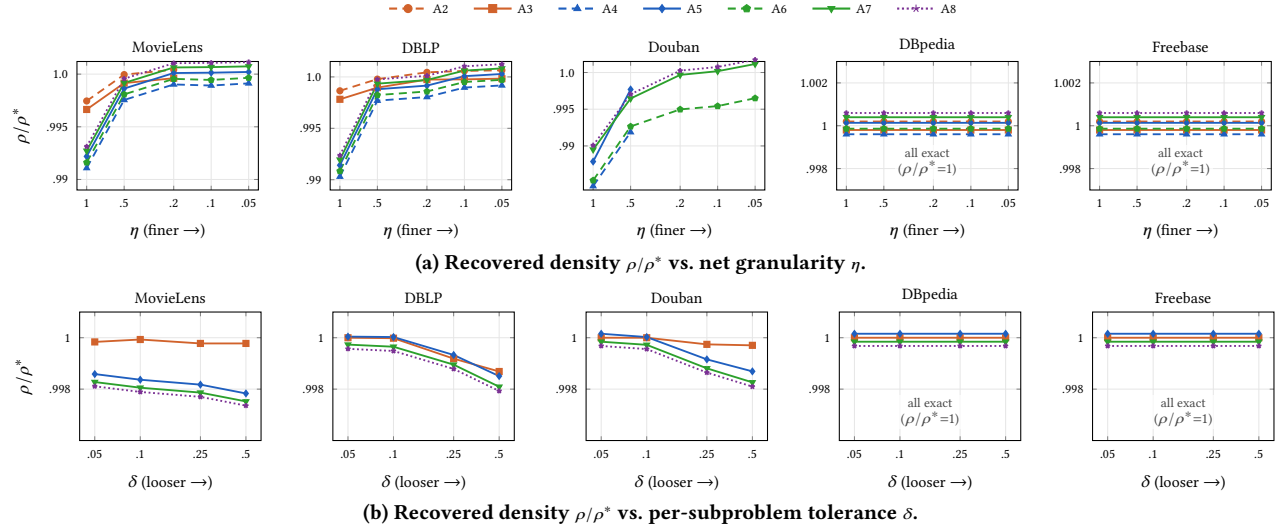

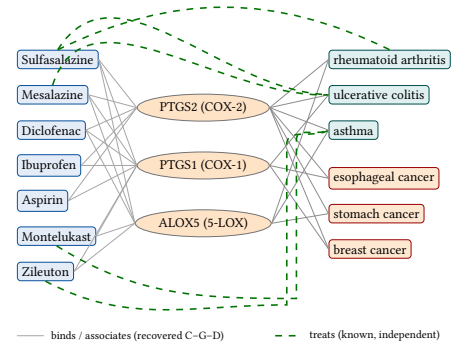
\begin{figure}[t]
\centering
\scalebox{0.85}{%
\begin{tikzpicture}[
  font=\scriptsize, >=Stealth,
  cmp/.style={draw=cS4, fill=cL3!12, rounded corners=1.5pt, inner sep=2pt,
              minimum height=8pt, anchor=west},
  gene/.style={draw=black!55, fill=orange!22, ellipse, inner sep=1.5pt, minimum height=9pt},
  inflam/.style={draw=teal!60!black, fill=teal!12, rounded corners=1.5pt, inner sep=2pt,
              minimum height=8pt, anchor=west},
  cancer/.style={draw=red!60!black, fill=orange!18, rounded corners=1.5pt, inner sep=2pt,
              minimum height=8pt, anchor=west},
  binds/.style={draw=black!35, line width=0.3pt},
  assoc/.style={draw=black!45, line width=0.3pt},
  treats/.style={draw=green!45!black, line width=0.7pt, dashed},
]
\node[cmp] (Su) at (0,3.30) {Sulfasalazine};
\node[cmp] (Me) at (0,2.75) {Mesalazine};
\node[cmp] (Di) at (0,2.20) {Diclofenac};
\node[cmp] (Ib) at (0,1.65) {Ibuprofen};
\node[cmp] (As) at (0,1.10) {Aspirin};
\node[cmp] (Mo) at (0,0.55) {Montelukast};
\node[cmp] (Zi) at (0,0.00) {Zileuton};
\node[gene] (P2) at (2.9,2.55) {PTGS2 (COX-2)};
\node[gene] (P1) at (2.9,1.65) {PTGS1 (COX-1)};
\node[gene] (AL) at (2.9,0.75) {ALOX5 (5-LOX)};
\node[inflam] (RA) at (4.85,3.30) {rheumatoid arthritis};
\node[inflam] (UC) at (4.85,2.75) {ulcerative colitis};
\node[inflam] (AS) at (4.85,2.20) {asthma};
\node[cancer] (EC) at (4.85,1.45) {esophageal cancer};
\node[cancer] (SC) at (4.85,0.90) {stomach cancer};
\node[cancer] (BC) at (4.85,0.35) {breast cancer};
\foreach \c/\g in {Su/P2,Su/P1,Su/AL, Me/P2,Me/P1,Me/AL, Di/P2,Di/P1,Di/AL,
                   Ib/P2,Ib/P1, As/P2,As/P1, Mo/P1,Mo/AL, Zi/P1,Zi/AL}
  \draw[binds] (\c.east) -- (\g.west);
\foreach \g/\d in {P2/RA,P2/UC,P2/AS,P2/EC,P2/SC,P2/BC, P1/UC,P1/EC,P1/BC,
                   AL/RA,AL/AS,AL/SC}
  \draw[assoc] (\g.east) -- (\d.west);
\draw[treats] (Su.north) to[bend left=26] (RA.north);
\draw[treats] (Su.north) to[out=58,in=178] (UC.west);
\draw[treats] (Me.north) to[out=72,in=182] (UC.west);
\draw[treats] (Mo.south) to[out=-32,in=180] (4.35,-0.42) -- (4.35,2.05)
                          to[out=90,in=208] (AS.west);
\draw[treats] (Zi.south) to[out=-28,in=180] (4.20,-0.52) -- (4.20,1.98)
                          to[out=90,in=212] (AS.west);
\begin{scope}[shift={(0,-1.0)}]
  \draw[binds] (0,0) -- (0.42,0); \node[anchor=west,font=\tiny] at (0.42,0) {binds / associates (recovered C--G--D)};
  \draw[treats] (4.02,0) -- (4.44,0); \node[anchor=west,font=\tiny] at (4.44,0) {treats (known, independent)};
\end{scope}
\end{tikzpicture}
}
\caption{Case study }
\label{fig:casestudy}
\end{figure}

\definecolor{cNetFirst}{HTML}{d1622b}
\definecolor{cNetSecond}{HTML}{1f5fbf}
\definecolor{cNetLocal}{HTML}{2ca02c}
\definecolor{cCertify}{HTML}{7b3294}
\newcommand{\cutline}[1]{\draw[densely dashed, black!55, line width=0.5pt]
  ({rel axis cs:0,0}|-{axis cs:0,#1}) -- ({rel axis cs:1,0}|-{axis cs:0,#1});}
\newcommand{\effscalCurves}[1]{%
  \addplot[black, densely dotted, mark=o,        mark options={solid}] table[x=frac,y=A1]{figures/data/#1};
  \addplot[cNetFirst, densely dashed, mark=*,    mark options={solid}] table[x=frac,y=A2]{figures/data/#1};
  \addplot[cNetFirst, solid, mark=square*]                             table[x=frac,y=A3]{figures/data/#1};
  \addplot[cNetSecond, densely dashed, mark=triangle*, mark options={solid}] table[x=frac,y=A4]{figures/data/#1};
  \addplot[cNetSecond, solid, mark=diamond*]                           table[x=frac,y=A5]{figures/data/#1};
  \addplot[cNetLocal, densely dashed, mark=pentagon*, mark options={solid}]  table[x=frac,y=A6]{figures/data/#1};
  \addplot[cNetLocal, solid, mark=triangle*, mark options={rotate=180}]      table[x=frac,y=A7]{figures/data/#1};
  \addplot[cCertify, densely dotted, mark=star, mark options={solid}]        table[x=frac,y=A8]{figures/data/#1};}
\newcommand{\effiDeltaCurves}[1]{%
  \addplot[cNetSecond, solid, mark=diamond*]
    table[x=delta,y=A5]{figures/data/effi_vary_delta_#1.dat};
  \addplot[cNetLocal, densely dashed, mark=pentagon*, mark options={solid}]
    table[x=delta,y=A6]{figures/data/effi_vary_delta_#1.dat};
  \addplot[cNetLocal, solid, mark=triangle*, mark options={rotate=180}]
    table[x=delta,y=A7]{figures/data/effi_vary_delta_#1.dat};
  \addplot[cCertify, densely dotted, mark=star, mark options={solid}]
    table[x=delta,y=A8]{figures/data/effi_vary_delta_#1.dat};}
\newcommand{\effiEtaCurves}[1]{%
  \addplot[cNetLocal, densely dashed, mark=pentagon*, mark options={solid}]
    table[x=eta,y=A6]{figures/data/effi_vary_eta_#1.dat};
  \addplot[cNetLocal, solid, mark=triangle*, mark options={rotate=180}]
    table[x=eta,y=A7]{figures/data/effi_vary_eta_#1.dat};
  \addplot[cCertify, densely dotted, mark=star, mark options={solid}]
    table[x=eta,y=A8]{figures/data/effi_vary_eta_#1.dat};}

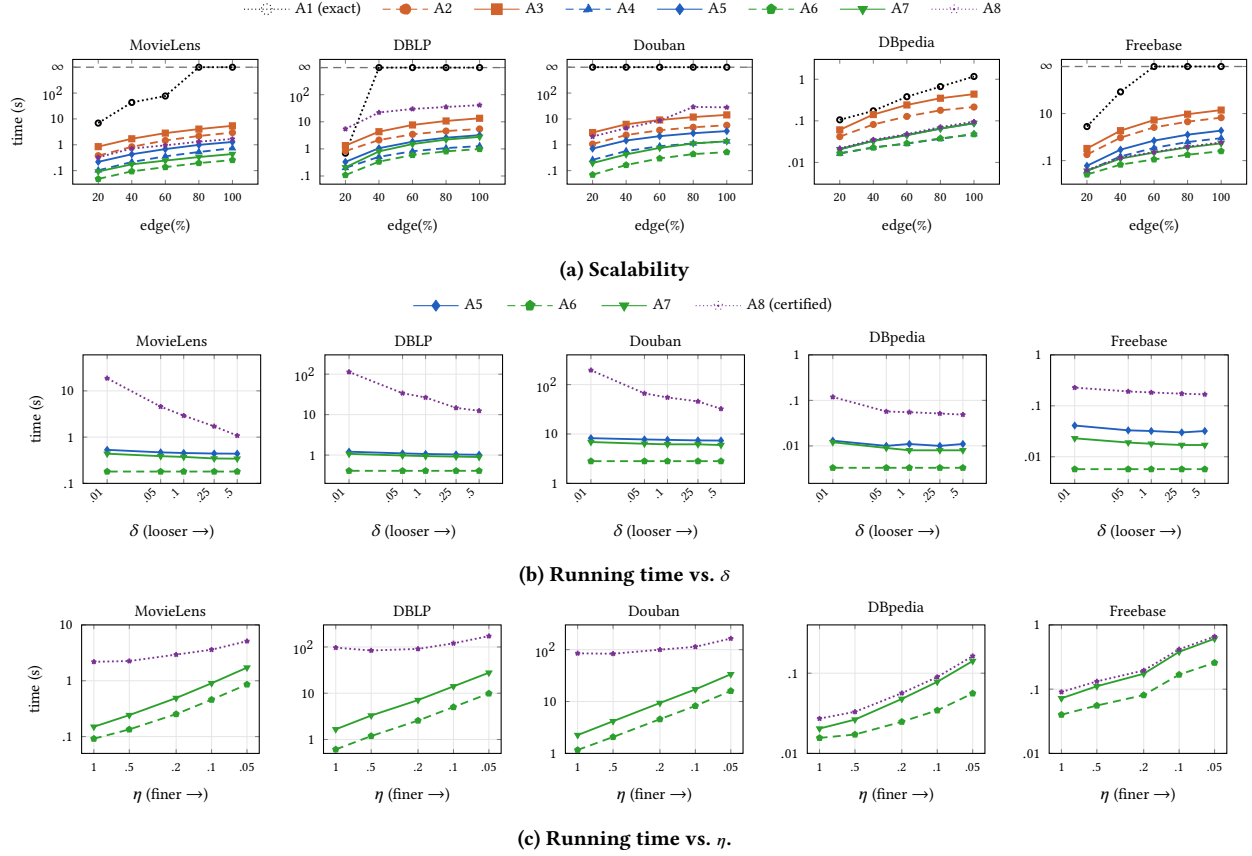
\begin{figure*}[t]
\centering
\begin{subfigure}{\textwidth}\centering
\begin{tikzpicture}
\begin{axis}[hide axis, scale only axis, width=1pt, height=1pt,
  legend columns=8, legend style={draw=none, font=\scriptsize,
  /tikz/every even column/.append style={column sep=6pt}}]
\addplot[draw=none,forget plot] coordinates {(0,0)};
\addlegendimage{black, densely dotted, mark=o}\addlegendentry{A1 (exact)}
\addlegendimage{cNetFirst, densely dashed, mark=*}\addlegendentry{A2}
\addlegendimage{cNetFirst, mark=square*}\addlegendentry{A3}
\addlegendimage{cNetSecond, densely dashed, mark=triangle*}\addlegendentry{A4}
\addlegendimage{cNetSecond, mark=diamond*}\addlegendentry{A5}
\addlegendimage{cNetLocal, densely dashed, mark=pentagon*}\addlegendentry{A6}
\addlegendimage{cNetLocal, mark=triangle*, mark options={rotate=180}}\addlegendentry{A7}
\addlegendimage{cCertify, densely dotted, mark=star}\addlegendentry{A8}
\end{axis}
\end{tikzpicture}\\[-0.2ex]
\begin{tikzpicture}
\begin{groupplot}[exprow,
    width=2.45cm, height=1.65cm,
    group style={group size=5 by 1, horizontal sep=0.82cm},
    ymode=log, log basis y=10,
    xmin=14, xmax=106,
    xtick={20,40,60,80,100}, xticklabels={20,40,60,80,100},
    x tick label style={font=\tiny},
    xlabel={edge(\%)},
    enlarge x limits=0.10,
    unbounded coords=jump,
  ]
\nextgroupplot[title={MovieLens}, ylabel={time (s)}, ymin=3e-2, ymax=2e3,
    ytick={0.1,1,10,100,1000}, yticklabels={$.1$,$1$,$10$,$10^2$,$\infty$}]
  \cutline{1000}
  \effscalCurves{effscal_MovieLens.dat}
\nextgroupplot[title={DBLP}, ymin=5e-2, ymax=2e3,
    ytick={0.1,1,10,100,1000}, yticklabels={$.1$,$1$,$10$,$10^2$,$\infty$}]
  \cutline{1000}
  \effscalCurves{effscal_DBLP.dat}
\nextgroupplot[title={Douban}, ymin=3e-2, ymax=2e3,
    ytick={0.1,1,10,100,1000}, yticklabels={$.1$,$1$,$10$,$10^2$,$\infty$}]
  \cutline{1000}
  \effscalCurves{effscal_Douban.dat}
\nextgroupplot[title={DBpedia}, ymin=3e-3, ymax=3,
    ytick={0.01,0.1,1}, yticklabels={$.01$,$.1$,$1$}]
  \effscalCurves{effscal_DBpedia.dat}
\nextgroupplot[title={Freebase}, ymin=1e-2, ymax=2e3,
    ytick={0.1,10,1000}, yticklabels={$.1$,$10$,$\infty$}]
  \cutline{1000}
  \effscalCurves{effscal_Freebase.dat}
\end{groupplot}
\end{tikzpicture}
\subcaption{Scalability 
}
\label{fig:exp-effscal}
\end{subfigure}\\[-0.4ex]
\begin{subfigure}{\textwidth}\centering
\begin{tikzpicture}
\begin{axis}[hide axis, scale only axis, width=1pt, height=1pt,
  legend columns=4, legend style={draw=none, font=\scriptsize,
  /tikz/every even column/.append style={column sep=7pt}}]
\addplot[draw=none,forget plot] coordinates {(0,0)};
\addlegendimage{cNetSecond, mark=diamond*}\addlegendentry{A5}
\addlegendimage{cNetLocal, densely dashed, mark=pentagon*}\addlegendentry{A6}
\addlegendimage{cNetLocal, mark=triangle*, mark options={rotate=180}}\addlegendentry{A7}
\addlegendimage{cCertify, densely dotted, mark=star}\addlegendentry{A8 (certified)}
\end{axis}
\end{tikzpicture}\\[-0.3ex]
\begin{tikzpicture}
\begin{groupplot}[exprow,
    width=2.35cm, height=1.7cm,
    group style={group size=5 by 1, horizontal sep=0.85cm},
    xmode=log, log basis x=10,
    xtick={0.01,0.05,0.1,0.25,0.5}, xticklabels={.01,.05,.1,.25,.5},
    xmin=7e-3, xmax=7e-1,
    x tick label style={font=\tiny, rotate=45, anchor=east}, enlarge x limits=0.08,
    ymode=log, log basis y=10, unbounded coords=jump,
    xlabel={$\delta$ (looser $\rightarrow$)},
    grid=both, grid style={black!10, line width=0.3pt},
  ]
\nextgroupplot[title=MovieLens, ylabel={time (s)}, ymin=0.1, ymax=6e1,
    ytick={0.1,1,10}, yticklabels={$.1$,$1$,$10$}]
  \effiDeltaCurves{MovieLens}
\nextgroupplot[title=DBLP, ymin=0.2, ymax=3e2,
    ytick={1,10,100}, yticklabels={$1$,$10$,$10^2$}]
  \effiDeltaCurves{DBLP}
\nextgroupplot[title=Douban, ymin=1, ymax=4e2,
    ytick={1,10,100}, yticklabels={$1$,$10$,$10^2$}]
  \effiDeltaCurves{Douban}
\nextgroupplot[title=DBpedia, ymin=1.5e-3, ymax=1,
    ytick={0.01,0.1,1}, yticklabels={$.01$,$.1$,$1$}]
  \effiDeltaCurves{DBpedia}
\nextgroupplot[title=Freebase, ymin=3e-3, ymax=1,
    ytick={0.01,0.1,1}, yticklabels={$.01$,$.1$,$1$}]
  \effiDeltaCurves{Freebase}
\end{groupplot}
\end{tikzpicture}
\subcaption{Running time vs. $\delta$ 
}
\label{fig:effi-vary-delta}
\end{subfigure}\\[-0.4ex]
\begin{subfigure}{\textwidth}\centering
\begin{tikzpicture}
\begin{groupplot}[exprow,
    width=2.35cm, height=1.7cm,
    group style={group size=5 by 1, horizontal sep=0.85cm},
    xmode=log, x dir=reverse, log basis x=10,
    xtick={0.05,0.1,0.2,0.5,1}, xticklabels={.05,.1,.2,.5,1},
    x tick label style={font=\tiny}, enlarge x limits=0.08,
    ymode=log, log basis y=10,
    xlabel={$\eta$ (finer $\rightarrow$)},
    grid=both, grid style={black!10, line width=0.3pt},
  ]
\nextgroupplot[title=MovieLens, ylabel={time (s)}, ymin=5e-2, ymax=1e1,
    ytick={0.1,1,10}, yticklabels={$.1$,$1$,$10$}]
  \effiEtaCurves{MovieLens}
\nextgroupplot[title=DBLP, ymin=0.5, ymax=3e2,
    ytick={1,10,100}, yticklabels={$1$,$10$,$10^2$}]
  \effiEtaCurves{DBLP}
\nextgroupplot[title=Douban, ymin=1, ymax=3e2,
    ytick={1,10,100}, yticklabels={$1$,$10$,$10^2$}]
  \effiEtaCurves{Douban}
\nextgroupplot[title=DBpedia, ymin=1e-2, ymax=4e-1,
    ytick={0.01,0.1}, yticklabels={$.01$,$.1$}]
  \effiEtaCurves{DBpedia}
\nextgroupplot[title=Freebase, ymin=1e-2, ymax=1,
    ytick={0.01,0.1,1}, yticklabels={$.01$,$.1$,$1$}]
  \effiEtaCurves{Freebase}
\end{groupplot}
\end{tikzpicture}
\subcaption{Running time vs. $\eta$.}
\label{fig:effi-vary-eta}
\end{subfigure}
\caption{Efficiency evaluation.}
\label{fig:efficiency}
\end{figure*}

\noindent\textbf{Datasets.}
We use five benchmark HINs (Table~\ref{tab:datasets}); \textsc{Hetionet} serves the case
study (\S\ref{ssec:exp-eff}). Three are dense
\emph{core} benchmarks: \textsc{MovieLens}, \textsc{DBLP}, and
\textsc{Douban}~\cite{chen2023densest}. Two are large \emph{breadth} graphs rebuilt from public dumps, so their sizes differ from~\cite{chen2023densest}:
\textsc{DBpedia} ($18.7$M edges) and \textsc{Freebase} ($433$M edges), used to test scale and the limits of the cover on long, sparse meta-paths.

\noindent\textbf{Queries (meta-paths).}
As in~\cite{chen2023densest}, from each schema we materialize a pool of meta-path subgraphs grouped
by length $i$: every available $\mpath$ of $i{\in}\{3,4\}$ on the core datasets,
and a sampled $i{\in}\{3,\dots,9\}$ pool on the breadth datasets ($140$ each). Sampling is at the path-instance level, so each \mpath-partite graph is connected and has \mpath-instances.

\noindent\textbf{Methods.}
We compare 8 algorithms, listed in Table~\ref{tab:algs}.
A1, the state-of-the-art~\cite{chen2023densest} exact algorithm, gives the exact optimum
$\rho^{*}$ by enumerating necessary feasible iRM-sets with pruning (its approximation is
consistently slower, checking every iRM-set\footnote{Therefore, we directly use the exact method as the baseline.}).
 \textsc{A2}--\textsc{A8} factor the cover into a \emph{net} (which
weight-sets are represented) and a per-cell \emph{solver}: \emph{second-order} =
the $\sqrt{2\eta/e}$ net (\S\ref{ssec:net}), \emph{local} = the $\Delta$-box
localization (\S\ref{ssec:deltabox}), \emph{near-optimal} = the supermodular
solver (\S\ref{sec:supermod}). The near-optimal solver runs in two modes: a \emph{conjectured} stop
(``conj.'': A3, A5, A7), a few passes reaching the near-optimal density in practice
without a per-run guarantee~\cite{boob2020flowless}, and a \emph{certified} stop
(``cert.'': A8), run to the dual-gap certificate of Proposition~\ref{prop:solver}
that proves $\tfrac{1-\delta}{1+\eta}$ on every cell~\cite{chekuri2022supermodular};
\S\ref{sec:implicit} discusses the two. A2--A8 all apply the pruning of
\S\ref{sec:pruning} by default.
Among A2 to A8, A6 is the best in terms of time complexity, while A8 is the best in terms of accuracy. 



\noindent\textbf{Implementation and measures.}
All algorithms are implemented in C++ and compiled with GCC~13.3 (\texttt{-O2}) on M5 Max. We measure the algorithm runtime as the total CPU time,
excluding the I/O cost of loading the graph into main memory; a cut-off time of
$800$\,s and a $48$\,GB memory cap are set, and a non-completion is denoted $\infty$.
Each reported number is averaged over the dataset's query pool.

\noindent\textbf{Parameters.}
We vary net granularity $\eta$, per-subproblem tolerance $\delta$, and meta-path length $\plen$, with defaults $0.1$, $0.05$, $3$ ($4$)\footnote{$3$ for simple-schema HINs}, respectively.

\subsection{Effectiveness}
\label{ssec:exp-eff}

\noindent Except for A1, which supplies the reference optimum $\rho^{*}$, only
results finishing within the cut-off time are shown.


\noindent\textbf{Recovered density.}
Figures~\ref{fig:exp-eta} and~\ref{fig:exp-delta} report the density actually
recovered, $\rho/\rho^{*}$, as $\eta$ and $\delta$ vary. The first observation is
that every variant recovers the optimum: on the core datasets
$\rho/\rho^{*}{\ge}0.999$ at the default $\eta{=}0.1$ (A7 within
$1.5{\times}10^{-3}$), and on the breadth datasets the near-optimal solver returns
the exact optimum at every setting. The second observation is that the accuracy is
insensitive to both parameters: A7 varies under $1\%$ across $\eta$ and under
$0.2\%$ across $\delta$. The reason is the quadratic net loss of
Theorem~\ref{thm:net}, whose gap stays tiny even at a coarse $\eta$; hence coverage
and quality decouple and a coarse $(\eta,\delta)$ incurs no loss. A faint
ordering A3$\,\gtrsim\,$A5$\,\gtrsim\,$A7 also appears across $\delta$: the
un-reduced first-order net represents more weight-sets, so its best cell lands
slightly closer to the optimum, but the gap ($\le0.2\%$) is within the $(1{+}\eta)$
tolerance.

\noindent\textbf{Achieving vs.\ certifying near-optimality (A6/A7 vs.\ A8).}
Whether a subproblem is peeled to the fixed $1/i$ cap (A6), to the heuristic
near-optimal stop (A7), or to the worst-case $(1{-}\delta)$ dual-gap certificate
(A8), all three recover the same density (Figure~\ref{fig:exp-eta}). The reason is
that on real inputs the peel reaches a near-optimal subgraph within a few passes, so
certifying it only \emph{proves} the result rather than changing it. Consequently,
certification yields a guarantee rather than accuracy: A6 and A7 achieve near-optimal
density but cannot certify it on a given input, whereas A7's end-to-end
$(1{-}\delta)/(1{+}\eta)$ bound, secured by A8's certificate, holds whatever the data
(\S\ref{ssec:bridge} shows an input where the uncertified $1/i$ peel falls short).
Therefore, A6/A7 is preferable when the input is trusted and A8 when a worst-case
guarantee is required.

\noindent\textbf{Case study: drug repurposing.}
We run the motivating compound--gene--disease repurposing query of
\S\ref{sec:intro} directly on \textsc{Hetionet}~\cite{himmelstein2017hetionet}
($47$K nodes, $2.25$M edges), where, as argued there, a densest
$\langle\textsf{C},\textsf{G},\textsf{D}\rangle$ block can serve as a repurposing
shortlist. To make the result displayable, we filter out a few generic high-degree
``housekeeping'' hub genes (the analogue of stop-words); the query then returns a
coherent module: a family of anti-inflammatory drugs, the shared enzymes they target (COX-2 and two
relatives, the targets of aspirin and ibuprofen), and the inflammatory diseases
those enzymes drive (Figure~\ref{fig:casestudy}). The result can be checked against
held-out labels: the query reads only the \textsc{binds}/\textsc{associates} edges
and never a treatment relation, yet five of the drug--disease pairs it implies
(e.g.\ sulfasalazine--rheumatoid arthritis
) are \emph{approved} indications stored
in a \emph{separate, held-out} \textsc{treats} edge type, so the dense structure
recovers approved indications it was never shown. The remaining implied pairs (the same
enzymes to three cancers, carrying \emph{no} \textsc{treats} edge) match a
documented but still-investigational therapy, plausible novel leads for a repurposing
pipeline. Our exact solver certifies the returned subgraph is optimal (edit

\subsection{Efficiency}
\label{ssec:exp-eff2}


\noindent\textbf{Scalability}.
Following the subgraph-generation of \cite{chen2023densest}, for each dataset we
retain fixed percentages ($20$--$100\%$) of the edges (10 sets for each percentage) and run the
algorithms on the induced subgraphs (Figure~\ref{fig:exp-effscal}); $\plen$ is set to
$3$ on the core datasets and to $4$ on the breadth datasets.
The first observation is that only the
localized-net variants A6, A7, and A8 scale to the complete dataset. The reason is
that their net is $n$-independent and each pass runs in $\OO(|E|)$ without
materializing the $\OO(\prod_j|V_j|)$ instances, so their cost grows only with the
graph. In contrast, removing the reductions degrades scalability: the
second-order net without localization (A4/A5), and far more so the un-localized
first-order nets (A2/A3), grow sharply with the graph. On \textsc{Freebase}, A3
reaches $14$\,s against A7's $0.5$\,s, a $28\times$ difference. A1, enumerating every iRM-set, is the least scalable, timing out on every core dataset and
on \textsc{Freebase}. The results
verify that both reductions, the $\sqrt{2\eta/e}$ net and the $\Delta$-box
localization, are critical for scalability; we therefore drop A1--A3 below.

\noindent\textbf{Achieving vs.\ certifying near-optimality (A7 vs.\ A8).}
Figure~\ref{fig:effi-vary-delta} varies the per-subproblem tolerance $\delta$. The
first observation is that \emph{achieving} near-optimal density is inexpensive and
$\delta$-insensitive: the heuristic peel (A5, A7) costs under $1.5\times$ as $\delta$
tightens three orders of magnitude, and the $1/i$ peel (A6) is $\delta$-independent,
both finishing within a few seconds. In contrast, \emph{certifying} near-optimality
(A8) is far more costly, and its cost is data-dependent. The reason is that the
certified time is governed by the number of \emph{near-tied} net cells: on
\textsc{MovieLens} and the breadth datasets a single dominant cell lets the incumbent
prune the rest, so certification stays within $40$s across all $\delta$; on
\textsc{DBLP} and \textsc{Douban}, many cells carry near-tied density and each must be
certified to convergence, so the running time increases sharply with small $\delta$.

\noindent\textbf{Varying net granularity $\eta$.}
Figure~\ref{fig:effi-vary-eta} varies $\eta$ on the full graph. The running time
rises smoothly as $\eta$ refines, matching the predicted
$\OO((\log n/\sqrt\eta)^{\plen-1})$ cell count, with A7 a roughly constant factor
above the $1/i$ peel A6. The reason is the per-cell solver overhead, which is the
only difference between the two. A8 rises even more slowly, $1.8$--$2.3\times$ under
the ${\approx}20\times$ cell growth, as its time goes to the few near-tied cells,
whose number barely changes with $\eta$.

\definecolor{cNetLocal}{HTML}{2ca02c}
\definecolor{cCertify}{HTML}{7b3294}
\newcommand{\effiVaryICurves}[1]{%
  \addplot[cNetLocal, densely dashed, mark=pentagon*, mark options={solid}]
    table[x=i,y=A6]{figures/data/effi_vary_i_#1.dat};
  \addplot[cNetLocal, solid, mark=triangle*, mark options={rotate=180}]
    table[x=i,y=A7]{figures/data/effi_vary_i_#1.dat};
  \addplot[cCertify, densely dotted, mark=star, mark options={solid}, yshift=2.4pt]
    table[x=i,y=A8]{figures/data/effi_vary_i_#1.dat};}

\begin{figure}[t]
\centering
\begin{tikzpicture}
\begin{axis}[hide axis, scale only axis, width=1pt, height=1pt,
  legend columns=3, legend style={draw=none, font=\scriptsize,
  /tikz/every even column/.append style={column sep=7pt}}]
\addplot[draw=none,forget plot] coordinates {(0,0)};
\addlegendimage{cNetLocal, densely dashed, mark=pentagon*}\addlegendentry{A6 ($1/i$ peel)}
\addlegendimage{cNetLocal, mark=triangle*, mark options={rotate=180}}\addlegendentry{A7 (near-opt)}
\addlegendimage{cCertify, densely dotted, mark=star}\addlegendentry{A8 (certified)}
\end{axis}
\end{tikzpicture}\\[-0.2ex]
\begin{tikzpicture}
\begin{groupplot}[exprow,
  width=2.62cm, height=1.62cm,
  group style={group size=2 by 1, horizontal sep=0.42cm, y descriptions at=edge left},
  ymode=log, log basis y=10,
  xmin=2.7, xmax=9.3, xtick={3,5,7,9}, enlarge x limits=0.06,
  x tick label style={font=\tiny},
  ymin=3e-3, ymax=3e2, ytick={1e-2,1e0,1e2}, yticklabels={$10^{-2}$,$1$,$10^2$},
  xlabel={meta-path length $\plen$},
  grid=both, grid style={black!10, line width=0.3pt},
]
\nextgroupplot[title=DBpedia, ylabel={runtime (s)}]
  \effiVaryICurves{DBpedia}
\nextgroupplot[title=Freebase]
  \effiVaryICurves{Freebase}
\end{groupplot}
\end{tikzpicture}
\vspace{-6pt}
\caption{Varying $\plen$ on the breadth datasets.}
\label{fig:effi-vary-i}
\end{figure}
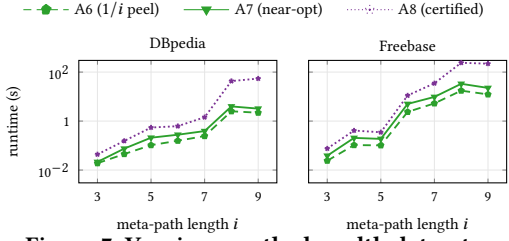

\noindent\textbf{Varying meta-path length $\plen$.}
Figure~\ref{fig:effi-vary-i} varies $\plen$ from $3$ up to $9$ on the breadth HINs,
averaging over the query graphs at each length. The first observation is
that only the localized variants remain tractable, their running time tracking the
net's $\OO((\log n/\sqrt\eta)^{\plen-1})$ growth, the one bottleneck the reductions
cannot remove. The second observation is that the heuristic A6 and A7 finish every
length (A7 within $33$\,s), whereas the certified A8 runs about an order of magnitude
slower when $i$ is large. 
This is consistent with the
achieve-vs-certify gap above, now amplified by the $\plen{-}1$ exponent in the net.

\newcommand{\ablbars}[4]{%
  \addplot[ybar,bar shift=0pt,fill=cOracle,draw=cS4,line width=0.3pt] coordinates {(0,#1)};
  \addplot[ybar,bar shift=0pt,fill=cS2,draw=cS4,line width=0.3pt]     coordinates {(1,#2)};
  \addplot[ybar,bar shift=0pt,fill=cS3,draw=cS4,line width=0.3pt]     coordinates {(2,#3)};
  \addplot[ybar,bar shift=0pt,fill=cL3,draw=cS4,line width=0.3pt]     coordinates {(3,#4)};}
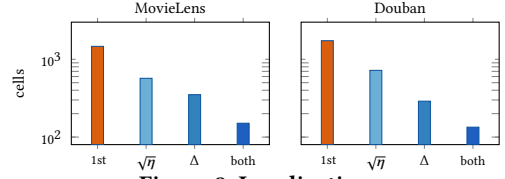
\begin{figure}[t]
\centering
\begin{tikzpicture}
\begin{groupplot}[exprow,
    group style={group size=2 by 1, horizontal sep=0.42cm, y descriptions at=edge left},
    ybar, /pgf/bar width=4.5pt,
    ymode=log, ymin=80, ymax=3000, log origin=infty,
    enlarge x limits=0.18,
    xtick={0,1,2,3},
    xticklabels={\,1st,$\sqrt\eta$,$\Delta$,both},
    x tick label style={font=\tiny,yshift=0.2ex},
    ytick={100,1000},
    yticklabels={$10^2$,$10^3$},
  ]
\nextgroupplot[title=MovieLens, ylabel={cells}]
  \ablbars{1467.09}{569.97}{350.76}{150.64}
\nextgroupplot[title=Douban]
  \ablbars{1733.15}{720.68}{289.48}{134.07}
\end{groupplot}
\end{tikzpicture}
\vspace{-6pt}
\caption{Localization
}
\label{fig:exp-ablation}
\end{figure}

\noindent\textbf{Effect of localization.}
Figure~\ref{fig:exp-ablation} reports the net cell count through the following
methods: 1) the first-order net (A2), 2) the $\sqrt\eta$ net (\S\ref{ssec:net}), 3) the
$\Delta$-box localization (\S\ref{ssec:deltabox}), and \textsf{both} (2 plus 3). \textsf{both} composes to an $\approx$$10\times$ smaller net on the core datasets, which
explains the fast running times of A6, A7, and A8 observed above.

%
%
%



\section{Related Work}\label{sec:related}

\noindent\textbf{Densest subgraph and its variants.}
The densest subgraph, which maximizes $|E(S)|/|S|$, is exact via parametric
max-flow~\cite{goldberg1984finding} and $1/2$-approximable by greedy
peeling~\cite{charikar2000greedy}; its
variants span directed/bipartite~\cite{khuller2009dense},
density-friendly~\cite{danisch2017large}, $k$-clique~\cite{tsourakakis2015kclique,he2023scaling},
multilayer~\cite{galimberti2017multilayer}, anchored~\cite{dai2022anchored}, and
locally densest clique/triangle~\cite{xu2024locally,yang2026locally} density, recently
consolidated in a unified analysis~\cite{zhou2024indepth} and a broad
survey~\cite{lanciano2024survey}. \textsf{DPSS} is their
$\plen$-partite meta-path generalization.

\noindent\textbf{Near-optimal supermodular solvers.}
A fast line recasts the objective as a \emph{supermodular} ratio solved to
$(1-\eps)$: Greedy++~\cite{boob2020flowless}, iterative
peeling~\cite{chekuri2022supermodular}, scalable
decomposition~\cite{harb2022faster} and its lexicographically-optimal-base
Frank--Wolfe view~\cite{harb2023lexbase}, coordinate descent~\cite{nguyen2024acdm}, and,
for the directed case, convex programming~\cite{ma2022convexdds}, scalable
discovery~\cite{zhou2025directed}, and accelerated coordinate
descent~\cite{liang2026racd}. All are one-dimensional, tied to the single
ratio at $\plen=2$. We lift them to $\plen>2$ via our weight-polytope cover
(\S\ref{sec:geometry},\,\S\ref{sec:supermod}). At $\plen=2$ they scan the ratio at a
first-order pitch ($\OO(\log n/\eps)$ guesses, each within $1{+}\eps$~\cite{ma2022convexdds});
our $\sqrt\eta$ net (\S\ref{ssec:net}) reaches the same $1{+}\eta$ loss in
$\OO(\log n/\sqrt\eta)$. For $\plen>2$ the feasible weights lose their total order
(\S\ref{sec:revisit}), so the ratio search no longer applies. 

\noindent\textbf{Dense subgraphs over HINs.}
HINs model typed, multi-relational data queried through
meta-paths~\cite{sun2012mining,shi2017survey,sun2011pathsim}; cohesive-structure
discovery there centers on query-anchored \emph{community
search}~\cite{fang2020survey,fang2020effective,jian2020effective,zhou2023influential},
with recent methods even avoiding meta-path \emph{materialization}~\cite{jiang2025scar},
closely related to our instance-free peeling. The
densest \mpath-partite subgraph~\cite{chen2023densest} optimizes one parameter-free
density globally; we improve this prior art.


\section{Conclusion}\label{sec:conclusion}

In this paper, we  show that the exhaustive weight-set enumeration behind prior \textsf{DPSS}
methods is unnecessary for approximation. The feasible iRM-sets fill a bounded
$(\plen-1)$-dimensional polytope that an $\eta$-net covers with
polylogarithmically many representatives. Each weighted supermodular
densest-subgraph instance is then solved near-optimally and instance-free, with
incumbent-driven pruning discarding most before applying any solver.
Overall, our algorithm achieves a $\tfrac{1-\delta}{1+\eta}$ guarantee
at polylogarithmic cost. 
Extensive experimental studies on real datasets justify the effectiveness and efficiency of our proposed methods. 

\bibliographystyle{ACM-Reference-Format}
\bibliography{sample-base}

\appendix
\setcounter{lemma}{0}
\setcounter{remark}{0}
\renewcommand{\thelemma}{A.\arabic{lemma}}
\renewcommand{\theremark}{A.\arabic{remark}}

\section{Full Proof of the Iterative Peeling Solver}\label{app:prop1}

\noindent
Throughout, $\irm=(m_1,\dots,m_{\plen})$ is a fixed iRM-set of per-type weights;
$f(S)=|\inst{S}|$ counts the instances of the length-$\plen$ meta-path \mpath within
$S\subseteq U$, the position-expanded ground set ($f$ is normalized, monotone, and
supermodular by Lemma~\ref{lem:supermodular}); $m_v:=m_{\phi'(v)}>0$ where $\phi'(v)\in[\plen]$
is $v$'s position in \mpath,
$w_{\irm}(S)=\sum_{v\in S}m_v$, $h_{\irm}(S)=f(S)/w_{\irm}(S)$ for $\emptyset\neq S\subseteq U$,
$h_{\irm}^{*}=\max_{\emptyset\neq S}h_{\irm}(S)$, $\cnt{v}{H}$ is the number of instances
through $v$ in a subgraph $H$, $\Delta_j=\max_{v\in V(A_j)}\cnt{v}{G}$
($V(A_j)$ the vertices of type $A_j$), and
$\Delta=\max_j\Delta_j$, the maximum per-vertex instance count. Deleting $v$ from $S$ destroys exactly the
instances through $v$, so, with $G(S)$ the subgraph induced by $S$,
$\cnt{v}{G(S)}=f(S)-f(S\setminus v)=f(v\,|\,S\setminus v)$. A vector
$x\in\mathbb R^{U}$ has one coordinate $x_v$ per element $v\in U$, and we write
$x(S)=\sum_{v\in S}x_v$; the (contra)polymatroid base polytope of $f$ is
\[
B(f)=\{x\ge0:x(S)\ge f(S)\ \forall S\subseteq U,\ x(U)=f(U)\}.
\]
For every $x\in B(f)$ and $v\in U$,
$f(\{v\})\le x_v=x(U)-x(U{-}v)\le f(U)-f(U{-}v)=f(v\,|\,U{-}v)=\cnt{v}{G}\le\Delta$,
so $B(f)\subseteq[0,\Delta]^{U}$ is compact and convex. The \emph{weighted width} is
$\Delta_{\irm}=\max_{v\in U}\cnt{v}{G}/m_v=\max_j\Delta_j/m_j$ ($m_v$ is constant on layers),
computable by the same two sweeps that build the $\Delta$-box
(Observations~\ref{lem:factorization}--\ref{thm:dp}).

\smallskip
\noindent\textbf{The algorithm in full.}
\textsc{Iterative Peeling}$(f,\irm,T)$: set $\sigma\equiv0$; for $t=1,\dots,T$: set
$S\leftarrow U$ and, while $S\neq\emptyset$, record $S$ as a candidate family, then delete
\[
\begin{gathered}
u=\arg\min_{v\in S}\frac{\sigma(v)+\cnt{v}{G(S)}}{m_v},\\[2pt]
g_t(u):=\cnt{u}{G(S)},\qquad \sigma(u)\mathrel{+}=g_t(u).
\end{gathered}
\]
Return the densest candidate family $\hat S$ over all prefixes (the surviving sets just before each
deletion) of all passes. At $T{=}1$, $\sigma\equiv0$ and the key is $\cnt{v}{G(S)}/m_v$, exactly
the reweighted peel of Theorem~\ref{thm:chen-peeling}.

\begin{lemma}[Greedy bases; exact linear minimization]\label{lem:greedybase}
Let $f$ be normalized, monotone, and supermodular on $U$. For an ordering
$\pi=(\pi_1,\dots,\pi_{|U|})$ of $U$ with prefixes $P_k=\{\pi_1,\dots,\pi_k\}$ define the
\emph{greedy base} $g^\pi$ by $g^\pi_{\pi_k}:=f(\pi_k\,|\,P_{k-1})$. Then:
(a) $g^\pi\in B(f)$ for every ordering $\pi$;
(b) for every $z\in\mathbb R^{U}_{\ge0}$, if $\pi$ sorts $z$ decreasingly
($z_{\pi_1}\ge\cdots\ge z_{\pi_{|U|}}$, ties arbitrary), then
\[
\min_{x\in B(f)}\langle z,x\rangle=\langle z,g^\pi\rangle
=\int_0^\infty f\big(\{v:z_v\ge\tau\}\big)\,d\tau=:\hat f(z),
\]
the Lov\'asz extension of $f$ at $z$.
\end{lemma}

\begin{proof}
(a) \emph{Nonnegativity:} $g^\pi_{\pi_k}=f(\pi_k\,|\,P_{k-1})\ge f(\{\pi_k\})\ge0$, since
supermodular marginals are nondecreasing in the context set and $f$ is monotone.
\emph{Tightness:} $g^\pi(U)=\sum_k\big(f(P_k)-f(P_{k-1})\big)=f(U)$ telescopes.
\emph{Covering:} fix $S\subseteq U$ and enumerate it in $\pi$-order as $s_1,\dots,s_r$ with
$s_j=\pi_{k_j}$; then $\{s_1,\dots,s_{j-1}\}\subseteq P_{k_j-1}$, so by monotone marginals
\[
f(S)=\sum_{j=1}^{r}f\big(s_j\,\big|\,\{s_1,\dots,s_{j-1}\}\big)
\le\sum_{j=1}^{r}f\big(s_j\,\big|\,P_{k_j-1}\big)=g^\pi(S).
\]
(b) Set $z_{\pi_{|U|+1}}:=0$. Abel summation gives, for any $x$,
$\langle z,x\rangle=\sum_{k}\big(z_{\pi_k}-z_{\pi_{k+1}}\big)x(P_k)$, with all coefficients
$\ge0$ by the sorting. For $x\in B(f)$, $x(P_k)\ge f(P_k)$ for $k<|U|$ and $x(P_{|U|})=f(U)$;
the greedy base attains $g^\pi(P_k)=f(P_k)$ for all $k$. Hence
$\langle z,x\rangle\ge\sum_k(z_{\pi_k}-z_{\pi_{k+1}})f(P_k)=\langle z,g^\pi\rangle$. Finally
$\{v:z_v\ge\tau\}=P_k$ for $\tau\in(z_{\pi_{k+1}},z_{\pi_k}]$, an interval of length
$z_{\pi_k}-z_{\pi_{k+1}}$, so the sum equals the integral defining $\hat f(z)$.
\end{proof}

\begin{lemma}[Weighted quasi-min-max]\label{lem:qmm}
With $f,\irm$ as above and $f\not\equiv0$,
\[
h_{\irm}^{*}=\max_{\emptyset\neq S\subseteq U}\frac{f(S)}{w_{\irm}(S)}
=\min_{x\in B(f)}\ \max_{v\in U}\ \frac{x_v}{m_v}.
\]
\end{lemma}

\begin{proof}
\emph{Weak direction ($\le$).} For $x\in B(f)$ and $\emptyset\neq S$:
$f(S)\le x(S)=\sum_{v\in S}\tfrac{x_v}{m_v}m_v\le\big(\max_{v\in S}\tfrac{x_v}{m_v}\big)w_{\irm}(S)$,
so $h_{\irm}(S)\le\max_{v\in U}x_v/m_v$ for every $S$ and every $x\in B(f)$.

\emph{Strong direction ($\ge$).} Let
$Z:=\{z\in\mathbb R^{U}_{\ge0}:\sum_v m_v z_v=1\}$, compact and convex since every $m_v>0$.
\emph{Step 1: $h_{\irm}^{*}=\max_{z\in Z}\hat f(z)$.} For ``$\le$'':
$z:=\mathbf 1_{S^{*}}/w_{\irm}(S^{*})\in Z$ for a maximizer $S^{*}$, and
$\hat f(z)=f(S^{*})/w_{\irm}(S^{*})=h_{\irm}^{*}$ by the integral formula of
Lemma~\ref{lem:greedybase}(b). For ``$\ge$'': take
any $z\in Z$ sorted decreasingly by $\pi$ with prefixes $P_k$ and coefficients
$\alpha_k:=z_{\pi_k}-z_{\pi_{k+1}}\ge0$; then $\hat f(z)=\sum_k\alpha_k f(P_k)$ and
$1=\sum_v m_v z_v=\sum_k\alpha_k w_{\irm}(P_k)$, so $f(P_k)\le h_{\irm}^{*}w_{\irm}(P_k)$ (each
$P_k\neq\emptyset$) gives $\hat f(z)\le h_{\irm}^{*}$.
\emph{Step 2: exchange min and max.} By Lemma~\ref{lem:greedybase}(b),
$\hat f(z)=\min_{x\in B(f)}\langle z,x\rangle$ on $Z\subseteq\mathbb R^{U}_{\ge0}$. The payoff
$\langle z,x\rangle$ is bilinear and both $Z$ and $B(f)$ are compact convex sets, so von
Neumann's minimax theorem gives
\[
h_{\irm}^{*}=\max_{z\in Z}\min_{x\in B(f)}\langle z,x\rangle
=\min_{x\in B(f)}\max_{z\in Z}\langle z,x\rangle
=\min_{x\in B(f)}\max_{v\in U}\frac{x_v}{m_v},
\]
the last step because a linear functional over the weighted simplex $Z$ is maximized at an
extreme point $z=\mathbf e_v/m_v$.
\end{proof}

\begin{lemma}[Per-pass base membership]\label{lem:passbase}
Consider one pass of \textsc{Iterative Peeling} (with \emph{any} deletion key), deleting
$u_1,u_2,\dots,u_{|U|}$ in order, where $S_k=\{u_k,\dots,u_{|U|}\}$ is the surviving set before
the $k$-th deletion, and charging $g(u_k)=\cnt{u_k}{G(S_k)}=f(u_k\,|\,S_k\setminus u_k)$. Then
$g=g^{\pi}$ for the \emph{reverse} deletion order $\pi=(u_{|U|},\dots,u_1)$; hence
$g\in B(f)$ by Lemma~\ref{lem:greedybase}(a). Consequently, after $t$ passes with charges
$g_1,\dots,g_t$ and accumulated load $\sigma_t:=\sum_{s\le t}g_s$, the averaged load
$\bar x_t:=\sigma_t/t$ lies in $B(f)$ by convexity.
\end{lemma}

\begin{proof}
Under $\pi=(u_{|U|},\dots,u_1)$ the element at position $j$ is $u_{|U|-j+1}$ and the prefix is
$P_{j-1}=\{u_{|U|},\dots,u_{|U|-j+2}\}=S_{|U|-j+1}\setminus u_{|U|-j+1}$; its greedy marginal is
therefore $f(u_{|U|-j+1}\,|\,S_{|U|-j+1}\setminus u_{|U|-j+1})=g(u_{|U|-j+1})$, coordinatewise
identical to the pass's charge.
\end{proof}

\noindent
The membership in Lemma~\ref{lem:passbase} holds for \emph{any} deletion key: every deletion
order emits some greedy base; the key only selects which one. No exact linear-minimization claim is made for
the load-augmented key: the reverse deletion order need not sort either the pre-pass or the
post-pass loads, which is why the convergence of
Greedy++ was long a conjecture~\cite{boob2020flowless}, resolved by a tailored
analysis~\cite{chekuri2022supermodular}.

\begin{proof}[Proof of Proposition~\ref{prop:solver}]
\emph{Part 1: Sound reporting via weak duality} (steps (i)--(ii) of the main-paper proof
sketch). By Lemma~\ref{lem:passbase}, after any $t$ passes
$\bar x_t=\sigma_t/t\in B(f)$, so by the weak direction of Lemma~\ref{lem:qmm},
\begin{equation}\label{eq:cert}
h_{\irm}^{*}\le\max_{v\in U}\frac{\bar x_t(v)}{m_v}
\quad\text{and}\quad
h_{\irm}(\hat S)\ge\max_{v}\frac{\bar x_t(v)}{m_v}-\hat g_t,
\end{equation}
where $\hat g_t:=\max_v\bar x_t(v)/m_v-\max_{\text{prefixes }S}h_{\irm}(S)$ is the nonnegative,
computable duality gap after $t$ passes. Thus $\hat g_t\le\delta\,h_{\irm}^{*}$ certifies
$h_{\irm}(\hat S)\ge(1-\delta)h_{\irm}^{*}$, unconditionally.

\emph{Part 2: The gap closes at the stated rate} (sketch step (iii)). It remains to bound $\min_{t\le T}\hat g_t$. We separate what is
proved here from what is imported, since only one step of the argument relies on an external result.

\emph{Correctness is self-contained.} By \eqref{eq:cert}, whenever the computable gap obeys
$\hat g_t\le\delta\,h_{\irm}^{*}$, the reported family satisfies
$h_{\irm}(\hat S)\ge(1-\delta)h_{\irm}^{*}$, and this uses \emph{only} weak duality
(Lemma~\ref{lem:qmm}) and averaged-base membership (Lemma~\ref{lem:passbase}), both established
above for the per-type-weighted objective. Thus every family the algorithm reports is certified
without appeal to any external rate (Remark~\ref{rem:cert}). What is \emph{imported}
from~\cite{chekuri2022supermodular} is solely the a~priori guarantee that this gap test is met
within the stated number of passes; we now show that this rate transfers to the weighted
objective once every per-element quantity is read per unit weight, rather than reproving it.

\emph{The imported rate.} We invoke the analysis behind Theorem~1.3
of~\cite{chekuri2022supermodular}: for a
normalized, monotone, supermodular $F$ on $n$ elements, \textsc{Super-Greedy++} (peel by
$\sigma(v)+F(v\,|\,S\setminus v)$, charge the marginal, report the best prefix) attains
$\min_{t\le T}\hat g_t^{\mathrm{card}}\le\varepsilon\lambda^{*}$ once
$T\ge C\,\Delta_F\ln n/(\varepsilon^{2}\lambda^{*})$, where $\Delta_F=\max_v F(v\,|\,V{-}v)$,
$\lambda^{*}=\max_{\emptyset\neq S}F(S)/|S|$, $C$ is an absolute constant, and
$\hat g_t^{\mathrm{card}}$ is the cardinality-denominator gap; the theorem is stated there as a
$(1-\varepsilon)$-approximation, and its proof establishes the duality-gap form quoted here.
(This is the rate whose tightness
\cite{chekuri2022supermodular} leave open; see Remark~\ref{rem:delta2}.)

\emph{Reading the analysis per unit weight.} Regard each element $v$ as carrying its weight
$m_v=m_{\phi'(v)}>0$; every per-element quantity is then read per unit weight, under the
correspondence
\begin{center}
\footnotesize
\setlength{\tabcolsep}{3pt}
\begin{tabular}{@{}lll@{}}
\toprule
role in~\cite{chekuri2022supermodular} & unweighted & weighted\\
\midrule
per-element load & $\sigma(v)$ & $\sigma(v)/m_v$\\
per-element marginal & $F(v\,|\,S{\setminus}v)$ & $\cnt{v}{G(S)}/m_v$\\
deletion key & $\sigma(v){+}F(v\,|\,S{\setminus}v)$ & $(\sigma(v){+}\cnt{v}{G(S)})/m_v$\\
denominator & $|S|$ & $w_{\irm}(S)=\sum_{v\in S}m_v$\\
duality anchor & $\min_{x\in B(F)}\max_v x_v$ & $\min_{x\in B(f)}\max_v x_v/m_v$\\
width & $\max_v F(v\,|\,V{-}v)$ & $\max_v\cnt{v}{G}/m_v$\\
$\log$ factor & $\ln n$ & $\ln|U|$\\
\bottomrule
\end{tabular}
\end{center}
the anchor being $\lambda^{*}$ (unweighted) and $h_{\irm}^{*}$ (weighted, Lemma~\ref{lem:qmm}),
the width $\Delta_F$ and $\Delta_{\irm}$. Under this correspondence the proof
of~\cite{chekuri2022supermodular} transfers step by step, its steps falling into two kinds.

\emph{Combinatorial steps are independent of the weights.} The peeling order, the prefix
accounting, and base membership never consult the weights. No matter which key chooses the vertex, a pass charges each deleted $v$
its marginal $\cnt{v}{G(S)}=f(v\,|\,S\setminus v)$, so in reverse deletion order it emits a greedy
base of $f$, and the averaged loads $\bar x_t=\sigma_t/t$ lie in $B(f)$; this is
Lemma~\ref{lem:passbase}, proved with no reference to $\irm$. Which prefixes are charged, and in
what order, is identical to the unweighted run; the weights enter only
afterward, dividing the value attached to a coordinate.

\emph{Analytic steps are the same statement under a change of functional.} By \eqref{eq:cert} our
goal is to certify $\max_v\bar x_t(v)/m_v\to h_{\irm}^{*}$, while \cite{chekuri2022supermodular}
certify $\max_v\bar x_t(v)\to\lambda^{*}$. These are one statement with the linear functional
$\max_v(\cdot)_v$ replaced by $\max_v(\cdot)_v/m_v$ and $\lambda^{*}$ replaced by
$h_{\irm}^{*}$. Lemma~\ref{lem:qmm} shows the latter is again a min--max over the \emph{same} base
polytope $B(f)$, of the same combinatorial type (the Lov\'asz extension of $f$ paired against the
weighted simplex $Z=\{z\ge0:\sum_v m_v z_v=1\}$; Lemma~\ref{lem:qmm}, Step~2). Hence every
inequality in~\cite{chekuri2022supermodular} relating $\max_v\bar x_t(v)$ to $\lambda^{*}$ has an
exact counterpart relating $\max_v\bar x_t(v)/m_v$ to $h_{\irm}^{*}$, and the counterpart is valid because
the three inputs those inequalities depend on are all shared and weight-free: the polytope $B(f)$,
the greedy/linear-minimization identity (Lemma~\ref{lem:greedybase}), and base membership
(Lemma~\ref{lem:passbase}). Only two quantities update: the coordinate range under the new functional,
$\Delta_F\to\Delta_{\irm}$, and the logarithmic factor, $\ln n\to\ln|U|$.

\emph{Where the $\ln|U|$ and the width enter.} Two ingredients of~\cite{chekuri2022supermodular}
fix these factors, and the rescaling touches each transparently. First, the regret term arises
from an entropy over the max-player's pure strategies; those strategies are the $|U|$ coordinate directions
$\mathbf e_v/m_v$ spanning $Z$, so their number is $|U|$ whatever the weights $m_v$; rescaling a weight rescales its
direction but adds no strategy, which is exactly why the logarithm is $\ln|U|$ and never
$\ln w_{\irm}(U)$. Second, the per-round quantities the analysis bounds are the per-unit-weight
values $x_v/m_v$, whose range is the width $\Delta_{\irm}$; the weights enter the rate only by
setting this range. The precise assembly of these two ingredients into the pass count
$\OO(\Delta_{\irm}\ln|U|/(\delta^{2}h_{\irm}^{*}))$ (linear in the width and in $1/h_{\irm}^{*}$,
quadratic in $1/\delta$) is the tailored analysis of~\cite{chekuri2022supermodular}, the same
argument now read per unit weight, and we import it as such; the
$1/\delta^{2}$ exponent proper to peeling-type schemes is discussed in Remark~\ref{rem:delta2}.

\emph{Conclusion.} Combining,
\[
\min_{t\le T}\hat g_t\le\delta\,h_{\irm}^{*}\quad\text{once}\quad
T=\OO\big(\Delta_{\irm}\ln|U|/(\delta^{2}h_{\irm}^{*})\big),
\]
and \eqref{eq:cert} converts this
into the claimed $(1-\delta)$ guarantee for the reported $\hat S$. We claim sufficiency of this
$T$, not tightness: \cite{chekuri2022supermodular} pose the precise rate of \textsc{Super-Greedy++}
as an open problem, and any sharpening transfers through the same correspondence untouched. Independently of
the imported rate, the reported family is certified by \eqref{eq:cert} whenever the computable gap
test is met (Remark~\ref{rem:cert}).

\emph{Part 3: Width on the localized net; the $T{=}1$ case} (completing sketch step (iii)).
Since $m_v$ is constant on layers,
$\Delta_{\irm}=\max_j\Delta_j/m_j$, an identity in the paper's existing parameters. The net
solved is the $\Delta$-box enlarged by one lattice step $s(\eta)=\sqrt{2\eta/e}$, at most $1$
since $\eta\le1$ (Definitions~\ref{def:lattice} and~\ref{def:net}), so its points satisfy $\ln m_j\ge-\ln\Delta-s(\eta)$, hence
$m_j\ge e^{-1}/\Delta$; with $\Delta_j\le\Delta$ this gives
$\Delta_{\irm}\le e\,\Delta^{2}=\OO(\Delta^{2})$ uniformly over net cells (the cell-exact
accounting at the optimal cell's representative is Remark~\ref{rem:budget}). Finally, at $T{=}1$,
$\sigma\equiv0$ and the key is $\cnt{v}{G(S)}/m_v$: pass~1 coincides exactly with the reweighted peel of
Theorem~\ref{thm:chen-peeling}, whose $\tfrac1\plen$ guarantee transfers to $h_{\irm}$ since the
weighted density of \S\ref{ssec:reduction} equals $\plen\,h_{\irm}$; the load carried across
passes is what lifts it to $(1-\delta)$.
\end{proof}

\begin{remark}[Computable certificate; early stopping]\label{rem:cert}
After any $t$ passes, both sides of \eqref{eq:cert} are computable in $\OO(|U|)$ from $\sigma_t$
and the incumbent prefix $\hat S_t$: if
$\max_{v\in U}\sigma_t(v)/(t\,m_v)\le h_{\irm}(\hat S_t)/(1-\delta)$, then
$h_{\irm}(\hat S_t)\ge(1-\delta)h_{\irm}^{*}$, \emph{unconditionally} (weak duality plus
Lemma~\ref{lem:passbase}; no imported rate is involved). Proposition~\ref{prop:solver}
guarantees the test is met within the stated $T$; in practice it is met far earlier
(\S\ref{sec:experiments}).
\end{remark}

\begin{remark}[Uniform pass budget: $T=\OO(\plen\,\Delta\ln|U|/\delta^{2})$]\label{rem:budget}
Theorem~\ref{thm:main} invokes the $(1-\delta)$ guarantee \emph{only at the optimal cell's
representative} $\irm'$; every other cell merely contributes a candidate family. At that cell
the width-to-value ratio collapses to existing parameters. Let $\pfam^{*}$ be the optimal
family, $\densopt$ its density, and $\irmopt$ its own iRM-set; by Lemma~\ref{lem:deltabox},
$m^{*}_j=\bar P_j/\densopt$ with $1\le\bar P_j\le\Delta_j$. At $\irmopt$ itself both quantities
are exact: $h^{*}_{\irmopt}=\densopt/\plen$ (attained at $\pfam^{*}$ by AM--GM tightness of the
induced iRM-set, while Lemma~\ref{lem:amgm} caps $h^{*}_{\irm}\le\densopt/\plen$ for every
feasible $\irm$) and
$\Delta_{\irmopt}=\max_j\Delta_j/m^{*}_j=\densopt\max_j\Delta_j/\bar P_j$, whence
$\Delta_{\irmopt}/h^{*}_{\irmopt}=\plen\max_j\Delta_j/\bar P_j$. The net representative inherits
this up to constants: Lemma~\ref{lem:round} gives $|\ln m'_j-\ln m^{*}_j|\le s(\eta)\le1$
coordinatewise, so $\Delta_{\irm'}\le e^{s(\eta)}\densopt\max_j\Delta_j/\bar P_j$, and
Lemma~\ref{lem:surrogate} gives $h^{*}_{\irm'}\ge\densopt/(\plen(1+\eta))$; the two occurrences
of $\densopt$ \emph{cancel}:
\[
\frac{\Delta_{\irm'}}{h^{*}_{\irm'}}
\le e^{s(\eta)}\,\plen(1+\eta)\max_{j\in[\plen]}\frac{\Delta_j}{\bar P_j}
\le e\,\plen(1+\eta)\,\Delta,
\]
using $\bar P_j\ge1$ and $\Delta_j\le\Delta$. Hence, absorbing $1+\eta\le2$ into the
constant, the uniform per-cell budget
$T=\OO(\plen\,\Delta\ln|U|/\delta^{2})$ meets Proposition~\ref{prop:solver}'s requirement at
$\irm'$, and Theorem~\ref{thm:main} holds unchanged with a budget in the paper's existing
parameters, with no $h_{\irm}^{*}$. The mechanism is that at the optimal cell the same $\densopt$ that
inflates the width (through $m^{*}_j=\bar P_j/\densopt$) also inflates the value, leaving the
ratio $\max_j\Delta_j/\bar P_j$ (the global maximum incidence over the optimum's per-layer average
incidence), which is near $1$ when the optimum contains no vertex of exceptionally high
incidence. Two qualifications are in order. First, the
$\Delta$-only form is a \emph{pipeline} fact: as a standalone per-cell statement,
Proposition~\ref{prop:solver} does need $h_{\irm}^{*}$, since an arbitrary net cell can
push $\Delta_{\irm}/h_{\irm}^{*}$ well beyond $\plen\Delta$; only at $\irm'$ does the cancellation
occur. Second, $\Delta$ is the incidence of the graph the peeler actually runs on, so the
pruning of \S\ref{sec:pruning} shrinks the budget further; cells may also stop early by
Remark~\ref{rem:cert}.
\end{remark}

\begin{remark}[On the $1/\delta^{2}$ exponent]\label{rem:delta2}
$\OO(1/\varepsilon)$ iteration counts in this line of work come from accelerated first-order
methods~\cite{harb2022faster}, which maintain \emph{per-edge load-sharing variables}; the
meta-path analogue materializes one variable per (instance, position) pair,
$\Theta(\plen\,|\inst{U}|)$ state, exactly the blow-up that instance-freeness (C3,
\S\ref{sec:implicit}) is designed to avoid, and that the $\dleft(v)\dright(v)$ factorization
(Observation~\ref{lem:factorization}) cannot compress, since those iterates are not constant on
instance classes. Peeling-type
schemes, whose state is $\OO(|U|)$, currently pay
$1/\delta^{2}$~\cite{chekuri2022supermodular}.
\end{remark}

\end{document}